\documentclass[10pt,twocolumn,twoside]{IEEEtran}
\usepackage{mathrsfs}
\usepackage{amsfonts}
\usepackage{bbm}
\usepackage{amsthm,amssymb}
\usepackage{multirow}
\usepackage{array}
\usepackage{graphicx}
\usepackage{caption}
\usepackage{algorithm}
\usepackage{algorithmic}
\usepackage{amsmath}
\usepackage[english]{babel}
\usepackage{color}
\usepackage{subfigure}
\usepackage{boxedminipage}
\usepackage{multicol}

\newtheorem{definition}{Definition}
\newtheorem{theorem}{Theorem}

\begin{document}

\title{Generating Searchable Public-Key Ciphertexts with Hidden Structures for Fast Keyword Search}

\author{
Peng Xu, \IEEEmembership{Member, IEEE,}
Qianhong Wu,  \IEEEmembership{Member, IEEE,} 
Wei Wang, \IEEEmembership{Member, IEEE,} 
\\ Willy Susilo, \IEEEmembership{Senior Member, IEEE,}
Josep Domingo-Ferrer, \IEEEmembership{Fellow, IEEE,}
Hai Jin, \IEEEmembership{Senior Member, IEEE}
\IEEEcompsocitemizethanks{\IEEEcompsocthanksitem P. Xu and H. Jin are with Services Computing Technology and System Lab, Cluster and Grid Computing Lab, School of Computer Science and Technology, Huazhong University of Science and Technology, Wuhan, China. E-Mail: \{xupeng, hjin\}@mail.hust.edu.cn.}
\IEEEcompsocitemizethanks{\IEEEcompsocthanksitem Q. Wu is with the School of Electronics and Information Engineering, Beihang Univerisity, Beijing, China, and with the State Key Laboratory of Information Security, Institute of Information Engineering, Chinese Academy of Sciences, Beijing, China. E-mail: qhwu@xidian.edu.cn.}
\IEEEcompsocitemizethanks{\IEEEcompsocthanksitem W. Wang is with Cyber-Physical-Social Systems Lab, School of Computer Science and Technology, Huazhong University of Science and Technology, Wuhan, China. E-Mail: viviawangww@gmail.com.}
\IEEEcompsocitemizethanks{\IEEEcompsocthanksitem W. Susilo is with Centre for Computer and Information Security Research, School of Computer Science and Software Engineering, University of Wollongong, Australia. E-Mail: wsusilo@uow.edu.au.}
\IEEEcompsocitemizethanks{\IEEEcompsocthanksitem J. Domingo-Ferrer is with Universitat Rovira i Virgili, Department of Computer Engineering and Mathematics, UNESCO Chair in Data Privacy, 43007, Tarragona, Catalonia. E-Mail: josep.domingo@urv.cat.}
}

\IEEEcompsoctitleabstractindextext{
\begin{abstract}
Existing semantically secure public-key searchable encryption schemes take search time linear with the total number of the ciphertexts. This makes retrieval from large-scale databases prohibitive. To alleviate this problem, this paper proposes \emph{Searchable Public-Key Ciphertexts with Hidden Structures} (SPCHS) for keyword search as fast as possible without sacrificing semantic security of the encrypted keywords. In SPCHS, all keyword-searchable ciphertexts are structured by hidden relations, and with the search trapdoor corresponding to a keyword, the minimum information of the relations is disclosed to a search algorithm as the guidance to find all matching ciphertexts efficiently. We construct a simple SPCHS scheme from scratch in which the ciphertexts have a hidden star-like structure. We prove our scheme to be semantically secure based on 
the decisional bilinear 
Diffie-Hellman assumption in the Random Oracle (RO) model. The search 
complexity of our scheme is dependent on \emph{the actual number of the ciphertexts containing the queried keyword}, rather than the number of all ciphertexts. Finally, we present a generic SPCHS construction from anonymous identity-based encryption and \emph{collision-free full-identity malleable} Identity-Based Key Encapsulation Mechanism (IBKEM) with anonymity. We illustrate two collision-free full-identity malleable IBKEM instances, which are semantically secure and anonymous, respectively, in the RO and standard models. The latter instance enables us to construct an SPCHS scheme with semantic security in the standard model.
\end{abstract}

\begin{IEEEkeywords} 
Public-key searchable encryption, semantic security, identity-based key encapsulation mechanism,
identity based encryption
\end{IEEEkeywords}}

\maketitle \IEEEdisplaynotcompsoctitleabstractindextext
\IEEEpeerreviewmaketitle

\section{Introduction}

\IEEEPARstart{P}{ublic-key} encryption with keyword search (PEKS), introduced by Boneh \emph{et al.} in \cite{BCO04}, has the advantage that anyone who knows the receiver's public key can upload keyword-searchable ciphertexts to a server. The receiver can delegate the keyword search to the server. More
specifically, each sender separately encrypts a file and its extracted keywords and sends the resulting ciphertexts to a server; when the receiver wants to retrieve the files containing a specific keyword, he delegates a keyword search trapdoor to the server; the server finds the encrypted files containing the queried keyword without knowing the original files or the keyword itself, and returns the corresponding encrypted files to the receiver; finally, the receiver decrypts these encrypted files\footnote{Since the encryption of the original files can be separately processed with an independent public-key encryption scheme as in \cite{BCO04}, we only describe the encryption of the keywords (unless otherwise
clearly stated in the paper).}. The authors of PEKS \cite{BCO04} also presented semantic security against chosen keyword attacks (SS-CKA) in the sense that the server cannot distinguish the ciphertexts of the keywords of its choice before observing the corresponding keyword search trapdoors. It seems an appropriate security notion, especially if the keyword space has no high min-entropy. Existing semantically secure PEKS schemes take search time linear with the total number of all ciphertexts. This makes retrieval from large-scale databases prohibitive. Therefore, more efficient searchable public-key encryption is crucial for practically deploying PEKS schemes.

One of the prominent works to accelerate the search over encrypted keywords in the public-key setting is deterministic
encryption introduced by Bellare \emph{et al}. in \cite{BBN07}. An encryption scheme is deterministic if the encryption algorithm is deterministic. Bellare \emph{et al}. \cite{BBN07} focus on enabling search over encrypted keywords to be
as efficient as the search for unencrypted keywords, such that a ciphertext containing a given keyword can be retrieved
in time complexity logarithmic in the total number of all ciphertexts. This is reasonable because the encrypted keywords
can form a tree-like structure when stored according to their binary values. However, deterministic encryption has two inherent limitations. First, keyword privacy can be guaranteed only
for keywords that are \emph{a priori} hard-to-guess by the adversary (\emph{i.e.}, keywords with high
min-entropy to the adversary); second, certain information of a message leaks inevitably via the ciphertext of
the keywords since the encryption is deterministic. Hence, deterministic encryption is only applicable in special scenarios.

\subsection{Our Motivation and Basic Ideas}

We are interested in providing highly efficient search performance without
sacrificing semantic security in PEKS. Observe that a keyword space is usually of no high min-entropy in many
scenarios. Semantic security is crucial to guarantee keyword privacy in such
applications. Thus the linear search complexity of existing schemes is the major obstacle to their adoption.
Unfortunately, the linear complexity seems to be inevitable because the server has to scan and test each ciphertext, due to the fact that these ciphertexts (corresponding to the same keyword or not) are indistinguishable to the server.

\begin{figure}
\centering
\includegraphics [width=0.48\textwidth,height=0.25\textheight]{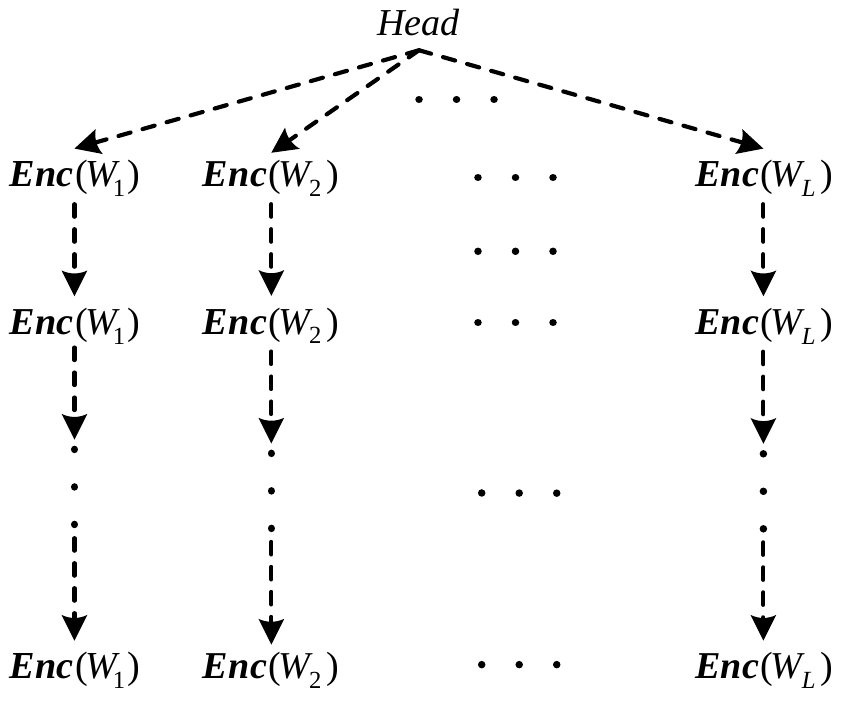}
\caption{Hidden star-like structure formed by keyword searchable ciphertexts. 
(The dashed arrows denote the hidden relations. $Enc(W_i)$ denotes the searchable ciphertext of keyword $W_i$.)}\label{F.Mainidea}
\end{figure}

A closer look shows that there is still space to improve search performance in PEKS without sacrificing semantic security if one can organize the ciphertexts with elegantly designed but hidden relations. Intuitively, if the keyword-searchable ciphertexts have a hidden star-like structure, as shown in Figure \ref{F.Mainidea}, then search over ciphertexts containing a specific keywords may be accelerated. Specifically, suppose all ciphertexts of the same keyword form a chain by the correlated hidden relations, and also a hidden relation exists from a public \emph{Head} to the first ciphertext of each chain. With a keyword search trapdoor and the \emph{Head}, the server seeks out the first matching ciphertext via the corresponding relation from the \emph{Head}. Then another relation can be disclosed via the found ciphertext and guides the searcher to seek out the next matching ciphertext. By carrying on in this way, all matching ciphertexts can be found. Clearly, the search time depends on the actual number of the ciphertexts containing the queried keyword, rather than on the total number of all ciphertexts.

To guarantee appropriate security, the hidden star-like structure should 
preserve the semantic security of keywords, which indicates that partial relations are disclosed only when the corresponding keyword search trapdoor is known. Each sender should be able to generate the keyword-searchable ciphertexts with the hidden star-like structure by the receiver's public-key; the server having a keyword search trapdoor should be able to disclose partial relations, which is related to all matching ciphertexts. Semantic security is preserved 1) if no keyword search trapdoor is
known, all ciphertexts are indistinguishable, and no information is leaked about the structure, and 2) given a keyword search trapdoor, only the corresponding relations can be disclosed, and the matching ciphertexts leak no information about the rest of ciphertexts, except the fact that the rest do not contain the queried keyword.

\subsection{Our Work}

We start by formally defining the concept of Searchable Public-key Ciphertexts with Hidden Structures (SPCHS) and its semantic security. In this new concept, keyword-searchable ciphertexts with their hidden structures can be generated in the public key setting; with a keyword search trapdoor, partial relations can be disclosed to guide the discovery of all matching ciphertexts. Semantic security is defined for both the keywords and the hidden structures. It is worth noting 
that this new concept and its semantic security are suitable for keyword-searchable ciphertexts with any kind of hidden structures. In contrast, the concept of traditional PEKS does not contain any hidden structure among the PEKS ciphertexts; correspondingly, its semantic security is only defined for the keywords. 

Following the SPCHS definition, we construct a simple SPCHS from scratch in the random oracle (RO) model. The scheme generates keyword-searchable ciphertexts with a hidden star-like structure. The search performance mainly depends on the actual number of the ciphertexts containing the queried keyword. For security, the scheme is proven semantically secure based on the Decisional Bilinear Diffie-Hellman (DBDH) assumption \cite{BB04} in the RO model.

We are also interested in providing a generic SPCHS construction to generate keyword-searchable ciphertexts with a hidden star-like structure. Our generic SPCHS is inspired by several interesting observations on Identity-Based Key Encapsulation Mechanism (IBKEM). In IBKEM, a sender encapsulates a key $K$ to an intended receiver $ID$. Of course, receiver $ID$ can decapsulate and obtain $K$, and the sender knows that receiver $ID$ will obtain $K$. However, a non-intended receiver $ID^\prime$ may also try to decapsulate and obtain $K^\prime$. We observe that, (1) it is usually the case that $K$ and $K^\prime$ are independent of each other from the view of the receivers, and (2) in some IBKEM the sender may also know $K^\prime$ obtained by receiver $ID^\prime$. We refer to the former property as \emph{collision freeness} and to the latter as \emph{full-identity malleability}. An IBKEM scheme is said to be \emph{collision-free full-identity malleable} if it possesses both properties. 

We build a generic SPCHS construction with Identity-Based Encryption (IBE) and collision-free full-identity malleable IBKEM. The resulting SPCHS can generate keyword-searchable ciphertexts with a hidden star-like structure. Moreover, if both the underlying IBKEM and IBE have semantic security and anonymity (\emph{i.e.} the privacy of receivers' identities), the resulting SPCHS is semantically secure. As there are known IBE schemes \cite{BW06,G06,AG09,D10} in both the RO model and the standard model, an SPCHS construction is reduced to collision-free full-identity malleable IBKEM with anonymity. In 2013, Abdalla \emph{et al.} proposed several IBKEM schemes to construct Verifiable Random Functions\footnote{VRF behaves like a pseudo-random function but one can verify that the output was
pseudo-random.} (VRF) \cite{ACF13}. We show that one of these IBKEM schemes is anonymous and collision-free full-identity malleable in the RO model. In \cite{FHPS13}, Freire \emph{et al.}  utilized the ``approximation'' of multilinear maps \cite{GGH13} to construct a standard-model version of Boneh-and-Franklin (BF) IBE scheme \cite{BF01}. We transform this IBE scheme into a collision-free full-identity malleable IBKEM scheme with semantic security and anonymity in the standard model. Hence, this new IBKEM scheme allows us to build SPCHS schemes secure in the standard model with the same search performance as the previous SPCHS construction 
from scratch 
in the RO model.

\subsection{Other Applications of Collision-Free Full-Identity Malleable IBKEM}

We note that collision-free full-identity malleable IBKEM is of independent interest. In addition to being a building block for the generic SPCHS construction, it may also find other applications, as outlined in the sequel.

\textbf{Batch identity-based key distribution.} A direct application of collision-free full-identity malleable IBKEM is to achieve batch identity-based key distribution. In such an application, a sender would like to distribute different secret session keys to multiple receivers so that each receiver can only know the session key to himself/herself. With collision-free full-identity malleable IBKEM, a sender just needs to broadcast an IBKEM encapsulation in the identity-based cryptography setting, e.g., encapsulating a session key $K$ to a single user $ID$. According to the collision-freeness of IBKEM, each receiver $ID'$ can decapsulate and obtain a different key $K'$ with his/her secret key in the identity based crypto-system. Due to the full-identity malleability, the sender knows the decapsulated keys of all the receivers. In this way, the sender efficiently shares different session keys with different receivers, at the cost of only a single encapsulation and one pass of communication.

\textbf{Anonymous identity-based broadcast encryption.}  A slightly more complicated application is anonymous identity-based broadcast encryption with efficient decryption. An analogous application was proposed respectively by Barth \emph{et al.} \cite{BBW06} and Libert \emph{et al.} \cite{LPQ12} in the traditional public-key setting. With collision-free full-identity malleable IBKEM, a sender generates an identity-based broadcast ciphertext $\langle C_1$, $C_2$, $(K_1^1||SE(K_2^1,F_1))$, $...$, $(K_1^N||SE(K_2^N,F_N))\rangle$,  where $C_1$ and $C_2$ are two IBKEM encapsulations, $K_1^i$ is the encapsulated key in $C_1$ for receiver $ID_i$, $K_2^i$ is the encapsulated key in $C_2$ for receiver $ID_i$, and $SE(K_2^i,F_i)$ is the symmetric-key encryption of file $F_i$ using the encapsulated key $K_2^i$. In this ciphertext, the encapsulated key $K_1^i$ is not used to encrypt anything. Indeed, it is an index to secretly inform receiver $ID_i$ on which part of this ciphertext belongs to him. To decrypt the encrypted file $F_i$, receiver $ID_i$ decapsulates and obtains $K_1^i$ from $C_1$, finds out $K_1^i||SE(K_2^i,F_i)$ by matching $K_1^i$, and finally extracts $F_i$ with the decapsulated key $K_2^i$ from $C_2$. It can be seen that the application will work if the IBKEM is collision-free full-identity malleable. It preserves the anonymity of receivers if the IBKEM is anonymous. Note that trivial anonymous broadcast encryption suffers decryption cost linear with the number of the receivers. In contrast, our anonymous identity-based broadcast encryption enjoys constant decryption cost, plus logarithmic complexity to search the matching index in a set $(K_1^1,...,K_1^N)$ organized by a certain partial order, e.g., a dictionary order according to their binary representations.

\subsection{Related Work}

Search on encrypted data has been extensively investigated in recent years. From a cryptographic perspective, the
existing works fall into two categories, {\em i.e.}, symmetric searchable encryption \cite{CGK06} and public-key
searchable encryption.

Symmetric searchable encryption is occasionally referred to as symmetric-key encryption with keyword search (SEKS).
This primitive was introduced by Song \emph{et al.} in \cite{SWP00}. Their instantiated scheme takes search time linear with the size of the database. A number of efforts \cite{G03,BC04,AKSX04,CM05,BCLN09} follow this research line and
refine Song \emph{et al.}'s original work. The SEKS scheme due to Curtmola \emph{et al}. \cite{CGK06} has been proven to be semantically secure against an
adaptive adversary. It allows the search to be processed in logarithmic time, although the keyword search trapdoor has
length linear with the size of the database. In addition to the above efforts devoted to either provable security or
better search performance, attention has recently been paid to achieving versatile SEKS schemes as follows. The works
in \cite{CGK06,BDDY08} extend SEKS to a multi-sender scenario. The work in \cite{LWWCRL10} realizes fuzzy keyword
search in the SEKS setting. The work in \cite{WBDS04} shows practical applications of SEKS and employs it to realize
secure and searchable audit logs. Chase \emph{et al.} \cite{CK10} proposed to encrypt structured data and a secure method to search these data. To support the dynamic update of the encrypted data, Kamara \emph{et al.}  proposed the dynamic searchable symmetric encryption in \cite{KPR12} and further enhanced its security in \cite{KP13} at the cost of large index. The very recent work \cite{CJJJKRS14} due to Cash \emph{et al.} simultaneously achieves strong security and high efficiency.

Following the seminal work on PEKS, Abdalla\emph{ et al}. \cite{ABC05} fills some gaps w.r.t. consistency for PEKS
and deals with the transformations among primitives related to PEKS. Some efforts have also
been devoted to make PEKS versatile. The work of this kind includes conjunctive search
\cite{PKL04,GSW04,BKM05,HL07,RT07,BSS08}, range search \cite{BCP06,SBC07,BW07}, subset search \cite{BW07}, time-scope search \cite{ABC05,DMR04}, similarity search \cite{CMWYZ10}, authorized search \cite{TC11,INH11}, equality test between heterogeneous ciphertexts \cite{YTHW10}, and fuzzy keyword search \cite{XHW13}. In addition, Arriaga et al. \cite{ATR14} proposed a PEKS scheme to keep the privacy of keyword search trapdoors.

In the above PEKS schemes, the search complexity takes time linear with the number of all ciphertexts. In [24],  an oblivious generation of keyword search trapdoor is to  maintain the privacy of the keyword against a curious trapdoor generation. A chain-like structure is described to speed up the search on encrypted keywords. One may note that the chain in \cite{CKRS09} cannot be fully hidden to the server and leaks the frequency of the keywords (see Supplemental Materials A for details). To realize an efficient keyword search, Bellare \emph{et al.} \cite{BBN07} introduced
deterministic public-key encryption (PKE) and formalized a security notion ``as strong as possible'' (stronger than
onewayness but weaker than semantic security). A deterministic searchable encryption scheme allows efficient keyword search as if the keywords were not
encrypted. Bellare \emph{et al.} \cite{BBN07} also presented a deterministic PKE scheme (\emph{i.e.}, RSA-DOAEP) and a
generic transformation from a randomized PKE to a deterministic PKE in the random oracle model. Subsequently,
deterministic PKE schemes secure in the standard model were independently proposed by Bellare \emph{et al.}
\cite{BFNR08} and Boldyreva \emph{et al.} \cite{BFN08}. The former uses general complexity assumptions and the
construction is generic, while the latter exploits concrete complexity assumptions and has better efficiency. Brakerski \emph{et al.} \cite{BS11} proposed the deterministic PKE schemes with better security, although these schemes are still not semantically secure. So far, deterministic PEKS schemes can guarantee semantic security only if the keyword space has a high
min-entropy. Otherwise, an adversary can extract the encrypted keyword by a simple encrypt-and-test attack. Hence, deterministic PEKS schemes are applicable to applications where the keyword space is of a high min-entropy. 

\subsection{Organization of this article}
The remaining sections are as follows. Section \ref{S.SPCHS.Concept} defines SPCHS and its semantic security. A simple SPCHS scheme is constructed in Section \ref{S.SPCHS.Instance}. A general construction of SPCHS is given in Section \ref{S.SPCHS.Generic}. Two collision-free full-identity malleable IBKEM schemes, respectively in the RO and standard models, are introduced in Section \ref{S.IBKEM.Instances}. Section \ref{S.Conclusion} concludes this paper. 

\section{Modeling SPCHS}\label{S.SPCHS.Concept}

We first explain intuitions behind SPCHS. We describe a hidden structure formed by ciphertexts as $(\mathbb{C},\mathbf{Pri},\mathbf{Pub})$, where $\mathbb{C}$ denotes the set of all ciphertexts, $\mathbf{Pri}$ denotes the hidden relations among $\mathbb{C}$, and $\mathbf{Pub}$ denotes the public parts. In case there is more than one hidden structure formed by ciphertexts, the description of multiple hidden structures formed by ciphertexts can be $(\mathbb{C},(\mathbf{Pri}_1,\mathbf{Pub}_1),...,(\mathbf{Pri}_N,\mathbf{Pub}_N))$, where $N\in\mathbb{N}$. Moreover, given $(\mathbb{C},\mathbf{Pub}_1,...,\mathbf{Pub}_N)$ and $(\mathbf{Pri}_1,...,\mathbf{Pri}_N)$ except $(\mathbf{Pri}_i,\mathbf{Pri}_j)$ (where $i\neq j$), one can neither learn anything about $(\mathbf{Pri}_i,\mathbf{Pri}_j)$ nor decide whether a ciphertext is associated with $\mathbf{Pub}_i$ or $\mathbf{Pub}_j$.

In SPCHS, the encryption algorithm has two functionalities. One is to encrypt a keyword, and the other is to generate a hidden relation, which can associate the generated ciphertext to the hidden structure. Let $(\mathbf{Pri},\mathbf{Pub})$ be the hidden structure. The encryption algorithm must take $\mathbf{Pri}$ as input, otherwise the hidden relation cannot be generated since $\mathbf{Pub}$ does not contain anything about the hidden relations. At the end of the encryption procedure, the $\mathbf{Pri}$ should be updated since a hidden relation is newly generated (but the specific method to update $\mathbf{Pri}$ relies on the specific instance of SPCHS). In addition, SPCHS needs an algorithm to initialize $(\mathbf{Pri},\mathbf{Pub})$ by taking the master public key as input, and this algorithm will be run before the first time to generate a ciphertext. With a keyword search trapdoor, the search algorithm of SPCHS can disclose partial relations
to guide the discovery of the ciphertexts containing the queried keyword with the hidden structure. 

\begin{definition}[SPCHS]\label{D.SPCHS.Concept}
SPCHS consists of five algorithms:
\begin{itemize}
\item $\mathbf{SystemSetup}(1^k,\mathcal{W})$: Take as input a security parameter $1^k$ and a keyword space $\mathcal{W}$, and probabilistically output a pair of master public-and-secret keys $(\mathbf{PK},\mathbf{SK})$, where $\mathbf{PK}$ includes the keyword space $\mathcal{W}$ and the ciphertext space $\mathcal{C}$.

\item $\mathbf{StructureInitialization}(\mathbf{PK})$: Take as input $\mathbf{PK}$, and probabilistically initialize a hidden structure by outputting its private and public parts $(\mathbf{Pri},\mathbf{Pub})$.

\item $\mathbf{StructuredEncryption}(\mathbf{PK},W,\mathbf{Pri})$: Take as inputs $\mathbf{PK}$, a keyword $W\in\mathcal{W}$ and a hidden structure's private part $\mathbf{Pri}$, and probabilistically output a keyword-searchable ciphertext $C$ of keyword $W$ with the hidden structure, and update $\mathbf{Pri}$.

\item $\mathbf{Trapdoor}(\mathbf{SK},W)$: Take as inputs $\mathbf{SK}$ and a keyword $W\in\mathcal{W}$, and output a keyword search trapdoor $T_W$ of $W$.

\item $\mathbf{StructuredSearch}(\mathbf{PK},\mathbf{Pub},\mathbb{C},T_W)$: Take as inputs $\mathbf{PK}$, a hidden structure's public part $\mathbf{Pub}$, all keyword-searchable ciphertexts $\mathbb{C}$ and a keyword search trapdoor $T_W$ of keyword $W$, disclose partial relations to guide finding out the ciphertexts containing keyword $W$ with the hidden structure.
\end{itemize}
An SPCHS scheme must be consistent in the sense that given any keyword search trapdoor $T_W$ and any hidden structure's public part $\mathbf{Pub}$, algorithm $\mathbf{StructuredSearch}(\mathbf{PK},\mathbf{Pub},\mathbb{C},T_W)$ finds out all ciphertexts of keyword $W$ with the hidden structure $\mathbf{Pub}$.
\end{definition}

In the application of SPCHS, a receiver runs algorithm $\mathbf{SystemSetup}$ to set up SPCHS. Each sender uploads the public part of his hidden structure and keyword-searchable ciphertexts to a server, respectively by algorithms $\mathbf{StructureInitialization}$ and $\mathbf{StructuredEncryption}$. Algorithm $\mathbf{Trapdoor}$ allows the receiver to delegate a keyword search trapdoor to the server. Then the server runs algorithm $\mathbf{StructuredSearch}$ for all senders' structures to find out the ciphertexts of the queried keyword.

The above SPCHS definition requires each sender to maintain the private part of his hidden structure for algorithm $\mathbf{StructuredEncryption}$. A similar requirement appears in symmetric-key encryption with keyword search (SEKS) in which each sender is required to maintain a secret key shared with the receiver. This implies interactions via authenticated confidential channels before a sender encrypts the keywords to the receiver in SEKS. In contrast,  each sender in SPCHS just generates and maintains his/her private values locally, i.e., without requirement of extra secure interactions before encrypting keywords. 

In the general case of SPCHS, each sender keeps his/her private values $\mathbf{Pri}$. We could let each sender be stateless by storing his/her $\mathbf{Pri}$ in encrypted form at a server and having each sender download and re-encrypt his/her $\mathbf{Pri}$ for each update of $\mathbf{Pri}$. A similar method also has been suggested by \cite{CJJJKRS14}.

The semantic security of SPCHS is to resist adaptively chosen keyword and structure attacks (SS-CKSA). In this security notion, a probabilistic polynomial-time (PPT) adversary is allowed to know all structures' public parts, query the trapdoors for adaptively chosen keywords, query the private parts of adaptively chosen structures, and query the ciphertexts of adaptively chosen keywords and structures (including the keywords and structures which the adversary would like to be challenged). The adversary will choose two challenge keyword-structure pairs. The SS-CKSA security means that for a ciphertext of one of two challenge keyword-structure pairs, the adversary cannot determine which challenge keyword or which challenge structure the challenge ciphertext corresponds to, provided that the adversary does not know the two challenge keywords' search trapdoors and the two challenge structures' private parts.

\begin{definition}[SS-CKSA Security]\label{D.SPCHS.Security}
Suppose there are at most $N\in\mathbb{N}$ hidden structures. An SPCHS scheme is SS-CKSA secure,
if any PPT adversary $\mathcal{A}$ has only a negligible advantage $Adv^{\text{SS-CKSA}}_{SPCHS,\mathcal{A}}$  to win in the following SS-CKSA game:

\begin{itemize}
\item \textbf{Setup Phase}: A challenger sets up the SPCHS scheme by running algorithm $\mathbf{SystemSetup}$ to generate a pair of master public-and-secret keys $(\mathbf{PK},\mathbf{SK})$, and initializes $N$ hidden structures by running algorithm $\mathbf{StructureInitialization}$ $N$ times (let $\mathbf{PSet}$ be the set of all public parts of these $N$ hidden structures.); finally the
challenger sends $\mathbf{PK}$ and $\mathbf{PSet}$ to $\mathcal{A}$.
\item \textbf{Query Phase 1}: $\mathcal{A}$ adaptively issues the following queries multiple times.
	\begin{itemize}
	 \item Trapdoor Query $\mathcal{Q}_{Trap}(W)$: Taking as input a keyword $W\in\mathcal{W}$, the challenger outputs the keyword search trapdoor of keyword $W$;
	\item Privacy Query $\mathcal{Q}_{Pri}(\mathbf{Pub})$: Taking as input a hidden structure's public part $\mathbf{Pub}\in\mathbf{PSet}$, the challenger outputs the corresponding private part of this structure;
	\item Encryption Query $\mathcal{Q}_{Enc}(W,\mathbf{Pub})$: Taking as inputs a keyword $W\in\mathcal{W}$ and a hidden structure's public part $\mathbf{Pub}$, the challenger outputs an SPCHS ciphertext of keyword $W$ with the hidden structure $\mathbf{Pub}$.
     \end{itemize}
\item \textbf{Challenge Phase}: $\mathcal{A}$ sends two challenge keyword-and-structure pairs $(W^*_0,\mathbf{Pub}^*_0)\in\mathcal{W}\times\mathbf{PSet}$ and $(W^*_1,\mathbf{Pub}^*_1)\in\mathcal{W}\times\mathbf{PSet}$ to the challenger; The challenger randomly chooses $d\in\{0,1\}$, and sends a challenge ciphertext $C^*_d$ of keyword $W_d^*$ with the hidden structure $\mathbf{Pub}_d^*$ to $\mathcal{A}$.
\item \textbf{Query Phase 2}: This phase is the same as \textbf{Query Phase 1}. Note that in \textbf{Query Phase 1} and \textbf{Query Phase 2}, $\mathcal{A}$ cannot query the corresponding private parts both of $\mathbf{Pub}^*_0$ and $\mathbf{Pub}^*_1$ and the keyword search trapdoors both of $W^*_0$ and $W^*_1$.
\item \textbf{Guess Phase}: $\mathcal{A}$ sends a guess $d^\prime$ to the challenger. We say that $\mathcal{A}$ wins if $d=d^\prime$. And let $Adv^{\text{SS-CKSA}}_{SPCHS,\mathcal{A}}=Pr[d=d^\prime]-\frac{1}{2}$ be the advantage of $\mathcal{A}$ to win in the above game.
\end{itemize}
\end{definition}

A weaker security definition of SPCHS is the selective-keyword security. We refer to this weaker security notion as SS-sK-CKSA security, and the corresponding attack game as SS-sK-CKSA game. In this attack game, the adversary $\mathcal{A}$ chooses two challenge keywords before the SPCHS scheme is set up, but the adversary still adaptively chooses two challenge hidden structures at $\textbf{Challenge Phase}$. Let $Adv^{\text{SS-sK-CKSA}}_{SPCHS,\mathcal{A}}$ denote the advantage of adversary $\mathcal{A}$ to win in this game.

\section{A Simple SPCHS Scheme from Scratch}\label{S.SPCHS.Instance}

Let $\gamma\overset{\$}\leftarrow \Re$ denote an element $\gamma$ randomly sampled from $\Re$. Let $\mathbb{G}$ and $\mathbb{G}_1$ denote two multiplicative groups of
prime order $q$. Let $g$ be a generator of $\mathbb{G}$. A bilinear map $\hat{e}:\mathbb{G}\times\mathbb{G} \rightarrow
\mathbb{G}_1$ \cite{MOV93,FMR99} is an efficiently computable and non-degenerate function, with the bilinearity
property $\hat{e}(g^a,g^b)=\hat{e}(g,g)^{ab}$, where $(a,b)\overset{\$} {\leftarrow} \mathbb{Z}_q^*$ and $\hat{e}(g,g)$
is a generator of $\mathbb{G}_1$. Let $\mathbf{BGen}(1^k)$ be an efficient bilinear map generator that takes as input a security
parameter $1^k$ and probabilistically outputs $(q,\mathbb{G},\mathbb{G}_1,g,\hat{e})$. Let keyword space $\mathcal{W}=\{0,1\}^*$. 

A simple SPCHS scheme secure in the random oracle model is constructed as follows.

\begin{itemize}
\item $\mathbf{SystemSetup}(1^k,\mathcal{W})$: Take as input $1^k$ and the keyword space $\mathcal{W}$, compute $(q,\mathbb{G}, \mathbb{G}_1, g,\hat{e})=\mathbf{BGen}(1^k)$, pick $s\overset{\$}\leftarrow \mathbb{Z}_q^*$, set $P=g^s$, choose a cryptographic hash function $H: \mathcal{W}\rightarrow \mathbb{G}$, set the ciphertext space $\mathcal{C}\subseteq\mathbb{G}_1\times\mathbb{G}\times\mathbb{G}_1$, and finally output the master public key $\mathbf{PK}=(q,\mathbb{G},\mathbb{G}_1,g,\hat{e},P,H,\mathcal{W},\mathcal{C})$, and the master secret key $\mathbf{SK}=s$. 

\item $\mathbf{StructureInitialization}(\mathbf{PK})$: Take as input $\mathbf{PK}$, pick $u\overset{\$}{\leftarrow}\mathbb{Z}_q^*$, and initialize a hidden structure by outputting a pair of private-and-public parts $(\mathbf{Pri}=(u),\mathbf{Pub}=g^u)$. Note that $\mathbf{Pri}$ here is a variable list formed as $(u,\{(W,Pt[u,W])|W\in\mathcal{W}\})$, which is initialized as $(u)$.

\item $\mathbf{StructuredEncryption}(\mathbf{PK},W,\mathbf{Pri})$: Take as inputs $\mathbf{PK}$, a keyword $W\in\mathcal{W}$, a hidden structure's private part $\mathbf{Pri}$, pick $r\overset{\$}{\leftarrow}\mathbb{Z}_q^*$ and do the following steps:
\begin{enumerate} 
\item Search $(W,Pt[u,W])$ for $W$ in $\mathbf{Pri}$;
\item If it is not found, insert $(W,Pt[u,W]\overset{\$}\leftarrow
\mathbb{G}_1)$ to $\mathbf{Pri}$, and output the keyword-searchable ciphertext $C=(\hat{e}(P,H(W))^u,g^r,\hat{e}(P,H(W))^r\cdot Pt[u,W])$;
\item Otherwise, pick $R\overset{\$}{\leftarrow}\mathbb{G}_1$, 
set $C=(Pt[u,W],g^r,\hat{e}(P,H(W))^r\cdot R)$, update $Pt[u,W]=R$, and output the keyword-searchable ciphertext $C$;
\end{enumerate}

\item $\mathbf{Trapdoor}(\mathbf{SK},W)$: Take as inputs $\mathbf{SK}$ and a keyword $W\in\mathcal{W}$, and output a keyword search trapdoor $T_W=H(W)^s$ of keyword $W$.

\item $\mathbf{StructuredSearch}(\mathbf{PK},\mathbf{Pub},\mathbb{C},T_W)$: Take as inputs $\mathbf{PK}$, a hidden structure's public part $\mathbf{Pub}$, all keyword-searchable ciphertexts $\mathbb{C}$ (let $\mathbb{C}[i]$ denote one ciphertext of $\mathbb{C}$, and this ciphertext can be parsed as $(\mathbb{C}[i,1],\mathbb{C}[i,2],\mathbb{C}[i,3])\in\mathbb{G}_1\times\mathbb{G}\times\mathbb{G}_1$) and a keyword trapdoor $T_W$ of keyword $W$, set $\mathbb{C}^\prime=\phi$, and do the following steps:
\begin{enumerate}
\item Compute $Pt^\prime=\hat{e}(\mathbf{Pub},T_W)$;
\item Seek a ciphertext $\mathbb{C}[i]$ having $\mathbb{C}[i,1]=Pt^\prime$; 
if it exists, add $\mathbb{C}[i]$ into $\mathbb{C}^\prime$; 
\item If no matching ciphertext is found, output $\mathbb{C}^\prime$;
\item Compute $Pt^\prime=\hat{e}(\mathbb{C}[i,2],T_W)^{-1}\cdot\mathbb{C}[i,3]$, and go to Step 2.
\end{enumerate}
\end{itemize}

Figure \ref{F.SPCHS.Instance} shows a hidden star-like structure, which is generated by the SPCHS instance. When running algorithm $\mathbf{StructuredSearch}(\mathbf{PK},\mathbf{Pub},\mathbb{C},T_{W_i})$, it discloses the value $\hat{e}(P,H(W_i))^u$ by computing $\hat{e}(\mathbf{Pub},T_{W_i})$,  and matches $\hat{e}(P,H(W_i))^u$ with all ciphertexts to find out the ciphertext $(\hat{e}(P,H(W_i))^u,g^r,\hat{e}(P,H(W_i))^r\cdot Pt[u,W_i])$. Then the algorithm discloses $Pt[u,W_i]$ by computing $\hat{e}(g^r,T_{W_i})^{-1}\cdot\hat{e}(P,H(W_i))^r\cdot Pt[u,W_i]$, and matches $Pt[u,W_i]$ with all ciphertexts to find out the ciphertext $(Pt[u,W_i],g^r,\hat{e}(P,H(W_i))^r\cdot R)$. 
By carrying on in this way,
the algorithm will find out all ciphertexts of keyword $W_i$ with the hidden star-like structure, and stop the search if no matching ciphertext is found. 

\begin{figure*}[!htp]
\centering\begin{boxedminipage}{0.9\textwidth}
\includegraphics[scale=0.7]{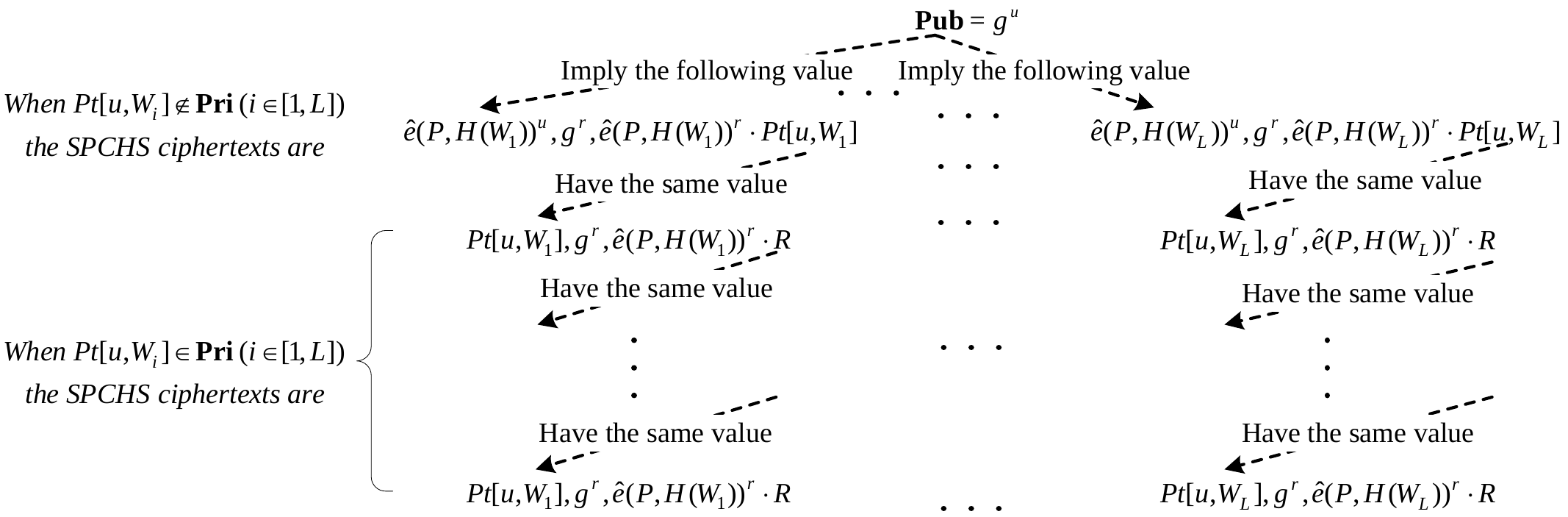}
\begin{quote}
Note that, in each ciphertext, the value $R$ and the value $r$ are randomly chosen. For $i\in[1,L]$, $Pt[u,W_i]$ is initialized with a random value when generating the first ciphertext of keyword $W_i$, and it will be updated into $R$ after generating each subsequent ciphertext of keyword $W_i$. 
\end{quote}  
\end{boxedminipage}
\caption{Hidden star-like structure generated by the above SPCHS instance}\label{F.SPCHS.Instance}
\end{figure*}

\textbf{Consistency.} Roughly speaking, algorithm $\mathbf{StructuredSearch}$ repetitively discloses the value of $Pt^\prime$ and matches the value with all ciphertexts' first parts to find out the matching ciphertexts. Since all disclosed values of $Pt^\prime$ are either collision-free (due to the hash function $H$) and random (according to algorithm $\mathbf{StructuredEncryption}$), no more than one ciphertext matches in each matching process. The found ciphertexts should contain the queried keyword, since given a keyword search trapdoor, algorithm $\mathbf{StructuredSearch}$ only can disclose the values of $Pt^\prime$, which are corresponding to the queried keyword. Formally, we have Theorem \ref{T.SPCHS.Instance.Consistency} on consistency whose proof can be found in Supplemental Materials B.

\begin{theorem}\label{T.SPCHS.Instance.Consistency}
Suppose the hash function $H$ is collision-free, except with a negligible probability in the security parameter $k$. The above SPCHS instance is consistent, also except with a negligible probability in the security parameter $k$.
\end{theorem}

\textbf{Semantic Security.} The SS-CKSA security of the above SPCHS scheme relies on the DBDH assumption in $\mathbf{BGen}(1^k)$. The definition of DBDH assumption \cite{BB04} is as follows.

\begin{definition}[The DBDH Assumption]\label{D.DBDH} The DBDH problem in $\mathbf{BGen}(1^k)=(q,\mathbb{G},\mathbb{G}_1,g,\hat{e})$ is defined as the advantage of any PPT algorithm $\mathcal{B}$ to distinguish the tuples $(g^a,g^b,g^c,\hat{e}(g,g)^{abc})$ and $(g^a,g^b,g^c,\hat{e}(g,g)^y)$, where $(a,b,c,y)\overset{\$}\leftarrow \mathbb{Z}_q^{*4}$. Let  $Adv_\mathcal {B}^{DBDH}(1^k)=\text{Pr}[\mathcal{B}(g^a,g^b,g^c,\hat{e}(g,g)^{abc})=1]-\text{Pr}[\mathcal{B}(g^a,g^b,g^c,\hat{e}(g,g)^y)=1]$ be the advantage of algorithm $\mathcal{B}$ to solve the DBDH problem. We say that the DBDH assumption holds in $\mathbf{BGen}(1^k)$, if the advantage $Adv_\mathcal {B}^{DBDH}(1^k)$ is negligible in the parameter $k$.
\end{definition}

In the security proof, we prove that if there is an adversary who can break the SS-CKSA security of the above SPCHS instance in the RO model, then there is an algorithm which can solve the DBDH problem in $\mathbf{BGen}(1^k)$. Formally we have Theorem \ref{T.SPCHS.Instance.Security} whose proof can be found in Supplemental Materials C.

\begin{theorem}\label{T.SPCHS.Instance.Security}
Let the hash function $H$ be modeled as the random oracle $\mathcal{Q}_{H}(\cdot)$. Suppose there are at most $N\in\mathbb{N}$ hidden structures, and a PPT adversary $\mathcal{A}$ wins in the SS-CKSA game of the above SPCHS instance with advantage $Adv^{\text{SS-CKSA}}_{SPCHS,\mathcal{A}}$, in which
$\mathcal{A}$ makes at most $q_t$ queries to oracle $\mathcal{Q}_{Trap}(\cdot)$ and at most $q_p$ queries to oracle $\mathcal{Q}_{Pri}(\cdot)$. Then there is a PPT algorithm $\mathcal{B}$ that solves the
DBDH problem in $\mathbf{BGen}(1^k)$ with advantage
$$Adv_\mathcal {B}^{DBDH}(1^k)\approx\frac{27}{(e\cdot q_t\cdot q_p)^3}\cdot Adv^{\text{SS-CKSA}}_{SPCHS,\mathcal{A}},$$ where $e$ is the base of natural logarithms.
\end{theorem}

\textbf{Forward and Backward Security.} Even in the case that a sender 
gets his local privacy $\mathbf{Pri}$ compromised, SPCHS still offers forward security. This means that the existing hidden structure of ciphertexts stays confidential, since the local privacy only contains the relationship of the new generated ciphertexts. To offer backward security with SPCHS, the sender can initialize a new structure by algorithm $\mathbf{StructureInitialization}$ for the new generated ciphertexts. Because the new structure is independent of the old one, the compromised local privacy will not leak the new generated structure.

\textbf{Search Complexity.} All keyword-searchable ciphertexts can be indexed by their first parts' binary bits. Assume that there are in total $n$ ciphertexts from $n_{\mbox{\tiny \emph S}}$ hidden structures, and the $i$-th hidden structure contains $n_{{\mbox{\tiny \emph W}},i}$ ciphertexts of keyword $W\in\mathcal{W}$. With the $i$-th hidden structure, the search complexity is $O(n_{{\mbox{\tiny \emph W}},i}\log n)$. For all hidden structures, the sum search complexity is $O((n_{\mbox{\tiny \emph S}}+n_{{\mbox{\tiny \emph W}}})\log n)$, where $n_{\mbox{\tiny \emph W}}=\sum_{i=1}^{n_{\mbox{\tiny \emph S}}} n_{{\mbox{\tiny \emph W}},i}$. Since $n=\sum_{W\in\mathcal{W}} n_{\mbox{\tiny \emph W}}$ and $n_{\mbox{\tiny \emph W}}=\sum_{i=1}^{n_{\mbox{\tiny \emph S}}} n_{{\mbox{\tiny \emph W}},i}$, we have that $n_{\mbox{\tiny \emph S}}\ll n_{\mbox{\tiny \emph W}}\ll n$. Thus the above SPCHS instance allows a much more efficient search than existing PEKS schemes, which have $O(n)$ search complexity.

One may note that SPCHS loses its significant advantage in search performance compared with PEKS if $n_{\mbox{\tiny \emph S}}=n$ holds. However, this special case seldom happens. In practice, a sender will extract several keywords from 
each of his files. So we usually have  $n_{\mbox{\tiny \emph S}}\ll n$ even if each sender only has one file. In addition, most related works on SEKS and PEKS assume that each file has several keywords.  

\textbf{Experiment.} We coded our SPCHS scheme, and tested the time cost of algorithm $\mathbf{StructuredSearch}$ to execute its cryptographic operations for different numbers of matching ciphertexts. We also coded the PEKS scheme \cite{BCO04}. Table \ref{T.Parameters} shows the system parameters including hardware, software and the chosen elliptic curve. Assume there are in total $10^4$ searchable ciphertexts. PEKS takes about 53.8 seconds search time per keyword, since it must test all ciphertexts for each search. 
Figure \ref{F.Result} shows the experimental results of SPCHS. It is clear that the time cost of SPCHS is linear with the number of matching ciphertexts, whereas for PEKS it is linear with 
the number of total ciphertexts. Hence, SPCHS is much more efficient than PEKS. 

\begin{table}
\centering
\caption{System parameters}\label{T.Parameters}
\begin{tabular}{|c|c|c|}
\hline
Hardware & \multicolumn{2}{c|}{Intel CPU
	E5300 @ 2.60GHz}
\\ \hline OS and compiler & \multicolumn{2}{c|}{Win XP and
Microsoft VC++ 6.0}
\\ \hline Program library & \multicolumn{2}{c|}{MIRACL version 5.4.1}
\\ \hline \multicolumn{3}{|c|}{Parameters of bilinear map}
\\ \hline Elliptic curve & \multicolumn{2}{c|}{$y^2=x^3+A\cdot x+B\cdot x$}
\\ \hline Pentanomial basis & \multicolumn{2}{c|}{$t^m+t^a+t^b+t^c+1$}
\\ \hline Base field: $2^m$ & \multicolumn{2}{c|}{$m=379$}
\\ \hline A & \multicolumn{2}{c|}{1}
\\ \hline B & \multicolumn{2}{c|}{1}
\\ \hline Group order: $q$ & \multicolumn{2}{c|}{$2^m+2^{(m+1)/2}+1$}
\\ \hline a & \multicolumn{2}{c|}{315}
\\ \hline b & \multicolumn{2}{c|}{301}
\\ \hline c & \multicolumn{2}{c|}{287}
\\ \hline \multicolumn{3}{|l|}{The default unit is decimal.}
\\ \hline
\end{tabular}
\end{table}

\begin{figure}
\centering
\includegraphics[width=0.5\textwidth,height=0.195\textheight]{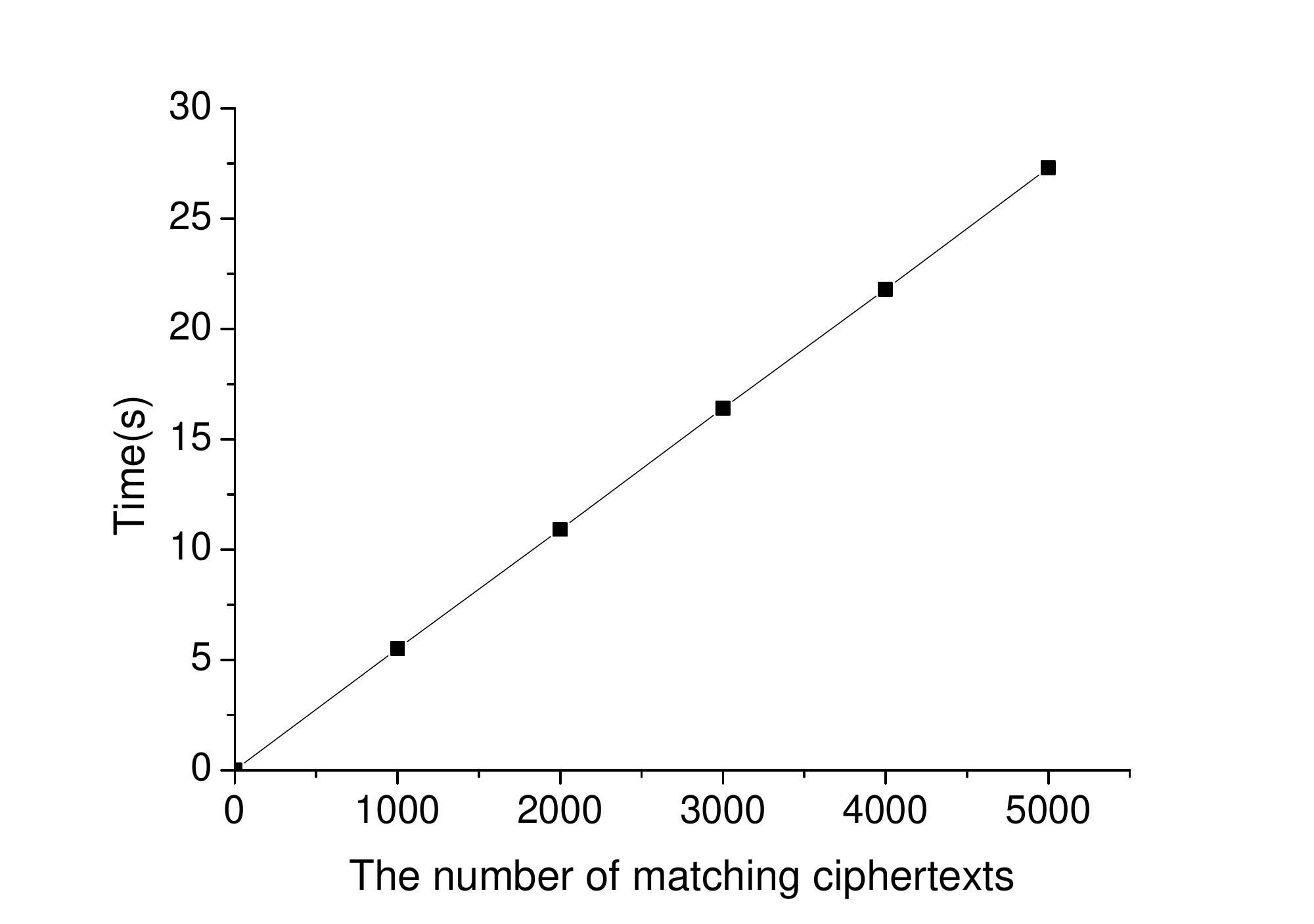}
\caption{Time cost of SPCHS}\label{F.Result}
\end{figure}

\section{A Generic Construction of SPCHS \\ from IBKEM and IBE}\label{S.SPCHS.Generic}

In this section, we formalize collision-free full-identity malleable IBKEM and a generic SPCHS construction from IBKEM and IBE. 

\subsection{Reviewing IBE}
Before the generic SPCHS construction, let us review the concept of IBE and its Anonymity and Semantic Security both under adaptive-ID and Chosen Plaintext Attacks (Anon-SS-ID-CPA).

\begin{definition}[IBE \cite{BF01}]\label{D.IBE}
IBE consists of four algorithms:
\begin{itemize}
\item $\mathbf{Setup}_{\mbox{\tiny IBE}}(1^k,\mathcal{ID}_{\mbox{\tiny IBE}})$: Take as inputs a security parameter $1^k$ and an identity space $\mathcal{ID}_{\mbox{\tiny IBE}}$, and probabilistically output the master public-and-secret-key pair $(\mathbf{PK}_{\mbox{\tiny IBE}},\mathbf{SK}_{\mbox{\tiny IBE}})$, where $\mathbf{PK}_{\mbox{\tiny IBE}}$ includes the message space $\mathcal{M}_{\mbox{\tiny IBE}}$, the ciphertext space $\mathcal{C}_{\mbox{\tiny IBE}}$ and the identity space $\mathcal{ID}_{\mbox{\tiny IBE}}$.

\item $\mathbf{Extract}_{\mbox{\tiny IBE}}(\mathbf{SK}_{\mbox{\tiny IBE}},ID)$: Take as inputs $\mathbf{SK}_{\mbox{\tiny IBE}}$ and an identity $ID\in\mathcal{ID}_{\mbox{\tiny IBE}}$, and output a decryption key $\tilde{S}_{ID}$ of $ID$.

\item $\mathbf{Enc}_{\mbox{\tiny IBE}}(\mathbf{PK}_{\mbox{\tiny IBE}},ID,M)$: Take as inputs $\mathbf{PK}_{\mbox{\tiny IBE}}$, an identity $ID\in\mathcal{ID}_{\mbox{\tiny IBE}}$ and a message $M$, and probabilistically output a ciphertext $\tilde{C}$.

\item $\mathbf{Dec}_{\mbox{\tiny IBE}}(\tilde{S}_{ID^\prime},\tilde{C})$: Take as inputs the decryption key $\tilde{S}_{ID^\prime}$ of identity $ID^\prime$ and a ciphertext $\tilde{C}$, and output a message or $\bot$, if the ciphertext is invalid.
\end{itemize}

An IBE scheme must be consistent in the sense that for any $\tilde{C}=\mathbf{Enc}_{\mbox{\tiny
IBE}}(\mathbf{PK}_{\mbox{\tiny IBE}},ID,M)$ and $\tilde{S}_{ID^\prime}=\mathbf{Extract}_{\mbox{\tiny IBE}}(\mathbf{SK}_{\mbox{\tiny
IBE}},ID^\prime)$, $\mathbf{Dec}_{\mbox{\tiny IBE}}(\tilde{S}_{ID^\prime},\tilde{C})=M$ holds if $ID^\prime=ID$, except with a negligible probability in the security parameter $k$.
\end{definition}

In the Anon-SS-ID-CPA security notion of IBE, a PPT adversary is allowed to query the decryption keys for adaptively chosen identities, and adaptively choose two challenge identity-and-message pairs. The Anon-SS-ID-CPA security of IBE means that for a challenge ciphertext, the adversary cannot determine which challenge identity and which challenge message it corresponds to, provided that the adversary does not know the two challenge identities' decryption keys. The Anon-SS-ID-CPA security of an IBE scheme is as follows. 

\begin{definition}[Anon-SS-ID-CPA security of IBE \cite{ABN10}] 
An IBE scheme is Anon-SS-ID-CPA secure if any PPT adversary $\mathcal{B}$ has only a negligible advantage $Adv^{\text{Anon-SS-ID-CPA}}_{IBE,\mathcal{B}}$ to win in the following Anon-SS-ID-CPA game:

\begin{itemize}
\item \textbf{Setup Phase}: A challenger sets up the IBE scheme by running algorithm $\mathbf{Setup}_{\mbox{\tiny IBE}}$ to generate the master public-and-secret-keys pair $(\mathbf{PK}_{\mbox{\tiny IBE}},\mathbf{SK}_{\mbox{\tiny IBE}})$, and sends $\mathbf{PK}_{\mbox{\tiny IBE}}$ to $\mathcal{B}$. 

\item \textbf{Query Phase 1}: Adversary $\mathcal{B}$ adaptively issues the following query multiple times.
\begin{itemize}
\item Decryption Key Query $\mathcal{Q}_{DK}^{IBE}(ID)$: Taking as input an identity $ID\in\mathcal{ID}_{\mbox{\tiny IBE}}$, the challenger outputs the decryption key of identity $ID$. 
\end{itemize}

\item \textbf{Challenge Phase}: Adversary $\mathcal{B}$ sends two challenge identity-and-message pairs $(ID_0^*,M_0^*)$ and $(ID_1^*,M_1^*)$ to the challenger; the challenger picks $\tilde{d}\overset{\$}\leftarrow\{0,1\}$, and sends the challenge IBE ciphertext $\tilde{C}^*_{\tilde{d}}=\mathbf{Enc}_{\mbox{\tiny IBE}}(\mathbf{PK}_{\mbox{\tiny IBE}},ID^*_{\tilde{d}},M^*_{\tilde{d}})$ to $\mathcal{B}$.

\item \textbf{Query Phase 2}: This phase is the same as \textbf{Query Phase 1}. Note that in \textbf{Query Phase 1} and \textbf{Query Phase 2}, $\mathcal{B}$ cannot query the decryption key corresponding to the challenge identity $ID_0^*$ or $ID_1^*$.

\item \textbf{Guess Phase}: Adversary $\mathcal{B}$ sends a guess $\tilde{d}^\prime$ to the challenger. We say that $\mathcal{B}$ wins if $\tilde{d}^\prime=\tilde{d}$. Let $Adv^{\text{Anon-SS-ID-CPA}}_{IBE,\mathcal{B}}=Pr[\tilde{d}^\prime=\tilde{d}]-\frac{1}{2}$ be the advantage of $\mathcal{B}$ to win in the above game. 
\end{itemize}
\end{definition}

\subsection{The Collision-Free Full-Identity Malleable IBKEM}

Our generic construction also relies on a notion of collision-free full-identity malleable IBKEM. The following IBKEM definition is derived from \cite{IP08}. A difference only appears in algorithm $\mathbf{Encaps}_{\mbox{\tiny IBKEM}}$. In order to highlight that the generator of an IBKEM encapsulation knows the chosen random value used in algorithm $\mathbf{Encaps}_{\mbox{\tiny IBKEM}}$, we take the random value as an input of the algorithm.

\begin{definition}[IBKEM]\label{D.IBKEM}
IBKEM consists of four algorithms:

\begin{itemize}
\item $\mathbf{Setup}_{\mbox{\tiny IBKEM}}(1^k,\mathcal{ID}_{\mbox{\tiny IBKEM}})$: Take as inputs a security parameter $1^k$ and an identity space $\mathcal{ID}_{\mbox{\tiny IBKEM}}$, and probabilistically output the master public-and-secret-keys pair $(\mathbf{PK}_{\mbox{\tiny IBKEM}},\mathbf{SK}_{\mbox{\tiny IBKEM}})$, where $\mathbf{PK}_{\mbox{\tiny IBKEM}}$ includes the identity space $\mathcal{ID}_{\mbox{\tiny IBKEM}}$, the encapsulated key space $\mathcal{K}_{\mbox{\tiny IBKEM}}$ and the encapsulation space $\mathcal{C}_{\mbox{\tiny IBKEM}}$.

\item $\mathbf{Extract}_{\mbox{\tiny IBKEM}}(\mathbf{SK}_{\mbox{\tiny IBKEM}},ID)$: Take as inputs $\mathbf{SK}_{\mbox{\tiny IBKEM}}$ and an identity $ID\in\mathcal{ID}_{\mbox{\tiny IBKEM}}$, and output a decryption key $\hat{S}_{ID}$ of $ID$.

\item $\mathbf{Encaps}_{\mbox{\tiny IBKEM}}(\mathbf{PK}_{\mbox{\tiny IBKEM}},ID,r)$: Take as inputs $\mathbf{PK}_{\mbox{\tiny IBKEM}}$, an identity $ID\in\mathcal{ID}_{\mbox{\tiny IBKEM}}$ and a random value $r$, and deterministically output a key-and-encapsulation pair $(\hat{K},\hat{C})$ of $ID$.

\item $\mathbf{Decaps}_{\mbox{\tiny IBKEM}}(\hat{S}_{ID^\prime},\hat{C})$: Take as inputs the decryption key $\hat{S}_{ID^\prime}$ of identity $ID^\prime$ and an encapsulation $\hat{C}$, and output an encapsulated key or $\bot$, if the encapsulation is invalid.
\end{itemize}

An IBKEM scheme must be consistent in the sense that for any $(\hat{K},\hat{C})=\mathbf{Encaps}_{\mbox{\tiny
IBKEM}}($ $\mathbf{PK}_{\mbox{\tiny IBKEM}},ID,r)$, $\mathbf{Decaps}_{\mbox{\tiny IBKEM}}(\hat{S}_{ID^\prime},\hat{C})=\hat{K}$ holds if $ID^\prime=ID$, except with a negligible probability in the security parameter $k$.
\end{definition}

The collision-free full-identity malleable IBKEM implies the following characteristics: all
identities' decryption keys can decapsulate the same encapsulation; all decapsulated keys are collision-free; the generator of the encapsulation can also compute these decapsulated keys; the decapsulated keys of different encapsulations are also collision-free.

\begin{definition}[Collision-Free Full-Identity Malleable IBKEM]\label{D.FIM}
IBKEM is collision-free full-identity malleable, if there is an efficient
function $\mathbf{FIM}$ that for any $(\hat{K},\hat{C})=\mathbf{Encaps}_{\mbox{\tiny IBKEM}}($ $\mathbf{PK}_{\mbox{\tiny IBKEM}},ID,r)$, the function $\mathbf{FIM}$ satisfies the following features:

\begin{itemize}
\item (Full-Identity Malleability) For any identity $ID^\prime\in\mathcal{ID}_{\mbox{\tiny IBKEM}}$, the equation $\mathbf{FIM}(ID^\prime,r)=\mathbf{Decaps}_{\mbox{\tiny IBKEM}}(\hat{S}_{ID^\prime},\hat{C})$ always holds, where
$\hat{S}_{ID^\prime}=\mathbf{Extract}_{\mbox{\tiny IBKEM}}(\mathbf{SK}_{\mbox{\tiny IBKEM}},ID^\prime)$;
\item (Collision-Freeness) For any identity $ID^\prime\in\mathcal{ID}_{\mbox{\tiny IBKEM}}$ and any random value $r^\prime$, if $ID\neq ID^\prime\bigvee r\neq r^\prime$, then $\mathbf{FIM}(ID,r)\neq \mathbf{FIM}(ID^\prime,r^\prime)$ holds,  except with a negligible probability in the security parameter $k$.
\end{itemize}
\end{definition}

A collision-free full-identity malleable IBKEM scheme may preserve semantic security and anonymity. We incorporate the semantic security and anonymity into Anon-SS-ID-CPA secure IBKEM. But this security is different from the traditional version  \cite{IP08} of the Anon-SS-ID-CPA security due to the full-identity malleability of IBKEM. The difference will be introduced after defining that security. In that security, a PPT adversary is allowed to query the decryption keys for adaptively chosen identities, and adaptively choose two challenge identities.  The Anon-SS-ID-CPA security of IBKEM means that for a challenge key-and-encapsulation pair, the adversary cannot determine the correctness of this pair and the challenge identity of this pair, given that the adversary does not know the two challenging identities' decryption keys. The Anon-SS-ID-CPA security of a collision-free full-identity malleable IBKEM scheme is as follows. 

\begin{definition}[Anon-SS-ID-CPA security of IBKEM]\label{D.IBKEM.Security}
An IBKEM scheme is Anon-SS-ID-CPA secure if any PPT adversary $\mathcal{B}$ has only a negligible advantage $Adv^{\text{Anon-SS-ID-CPA}}_{IBKEM,\mathcal{B}}$ to win in the following Anon-SS-ID-CPA game:

\begin{itemize}
\item \textbf{Setup Phase}: A challenger sets up the IBKEM scheme by running algorithm $\mathbf{Setup}_{\mbox{\tiny IBKEM}}$ to generate the master public-and-secret-keys pair $(\mathbf{PK}_{\mbox{\tiny IBKEM}},\mathbf{SK}_{\mbox{\tiny IBKEM}})$, and sends $\mathbf{PK}_{\mbox{\tiny IBKEM}}$ to $\mathcal{B}$. 

\item \textbf{Query Phase 1}: $\mathcal{B}$ adaptively issues the following query multiple times.
\begin{itemize}
\item Decryption Key Query $\mathcal{Q}_{DK}^{IBKEM}(ID)$: Taking as input an identity $ID\in\mathcal{ID}_{\mbox{\tiny IBKEM}}$, the challenger outputs the decryption key of identity $ID$. 
\end{itemize}

\item \textbf{Challenge Phase}: $\mathcal{B}$ sends two challenge identities $ID^*_0$ and $ID_1^*$ to the challenger; the challenger picks $\hat{d}\overset{\$}\leftarrow\{0,1\}$, computes $(\hat{K}^*_0,\hat{C}^*_0)=\mathbf{Encaps}_{\mbox{\tiny IBKEM}}(PK_{\mbox{\tiny IBKEM}},ID^*_0,r_0)$
and $(\hat{K}^*_1,\hat{C}^*_1)=\mathbf{Encaps}_{\mbox{\tiny IBKEM}}(PK_{\mbox{\tiny IBKEM}},ID^*_1,r_1)$, and sends the challenge key-and-encapsulation pair $(\hat{K}_{\hat{d}}^*,\hat{C}_0^*)$ to $\mathcal{B}$, where $r_0$ and $r_1$ are randomly chosen.

\item \textbf{Query Phase 2}: This phase is the same as \textbf{Query Phase 1}. Note that in \textbf{Query Phase 1} and \textbf{Query Phase 2}, $\mathcal{B}$ cannot query the decryption keys both of the challenge identities $ID_0^*$ and $ID_1^*$.

\item \textbf{Guess Phase}: $\mathcal{B}$ sends a guess $\hat{d}^\prime$ to the challenger. We say that $\mathcal{B}$ wins if $\hat{d}^\prime=\hat{d}$. Let $Adv^{\text{Anon-SS-ID-CPA}}_{IBKEM,\mathcal{B}}=Pr[\hat{d}^\prime=\hat{d}]-\frac{1}{2}$ be the advantage of $\mathcal{B}$ to win in the above game. 
\end{itemize}
\end{definition}

In the above definition, the anonymity of the encapsulated keys is defined by the indistinguishability of $\hat{K}_0^*$ and $\hat{K}_1^*$. But we do not define the anonymity of the IBKEM encapsulations (\emph{i.e.} the challenge key-and-encapsulation pair consists of $\hat{C}_0^*$ instead of $\hat{C}_{\hat{d}}^*$), since the full-identity malleability of IBKEM implies that any IBKEM encapsulation is valid for all identities.

A weaker security definition of IBKEM is the selective-identity security, referred to as the Anon-SS-sID-CPA security. The corresponding attack game is called the Anon-SS-sID-CPA game in which the adversary must commit to the two challenge identities before the system is set up.

\subsection{The Proposed Generic SPCHS Construction}

Let keyword space $\mathcal{W}\subset\mathcal{ID}_{\mbox{\tiny IBKEM}}=\mathcal{ID}_{\mbox{\tiny IBE}}$. Our generic SPCHS construction from the collision-free full-identity malleable IBKEM and IBE is as follows.

\begin{itemize}
\item $\mathbf{SystemSetup}(1^k,\mathcal{W})$: Take as inputs a security parameter $1^k$ and the keyword space $\mathcal{W}$, run $(\mathbf{PK}_{\mbox{\tiny IBKEM}},\mathbf{SK}_{\mbox{\tiny IBKEM}})=\mathbf{Setup}_{\mbox{\tiny IBKEM}}(1^k,\mathcal{ID}_{\mbox{\tiny IBKEM}})$ and $(\mathbf{PK}_{\mbox{\tiny IBE}},\mathbf{SK}_{\mbox{\tiny IBE}})=\mathbf{Setup}_{\mbox{\tiny IBE}}(1^k,$ $\mathcal{ID}_{\mbox{\tiny IBE}})$, and output a pair of master public-and-secret keys $(\mathbf{PK}=(\mathbf{PK}_{\mbox{\tiny IBKEM}},\mathbf{PK}_{\mbox{\tiny IBE}}),\mathbf{SK}=(\mathbf{SK}_{\mbox{\tiny IBKEM}},\mathbf{SK}_{\mbox{\tiny IBE}}))$. Let the SPCHS ciphertext space $\mathcal{C}=\mathcal{K}_{\mbox{\tiny IBKEM}}\times\mathcal{C}_{\mbox{\tiny IBE}}$, and $\mathcal{K}_{\mbox{\tiny IBKEM}}=\mathcal{M}_{\mbox{\tiny IBE}}$.

\item $\mathbf{StructureInitialization}(\mathbf{PK})$: Take as input $\mathbf{PK}$, arbitrarily pick a keyword $W\in\mathcal{W}$ and a random value $u$, generate an IBKEM encapsulated key and its encapsulation $(\hat{K},\hat{C})=\mathbf{Encaps}_{\mbox{\tiny IBKEM}}(\mathbf{PK}_{\mbox{\tiny IBKEM}},W,u)$, and initialize a hidden structure by outputting a pair of private-and-public parts $(\mathbf{Pri}=(u),\mathbf{Pub}=\hat{C})$. Note that $\mathbf{Pri}$ here is a variable list formed as $(u,\{(W,Pt[u,W])|W\in\mathcal{W}\})$, which is initialized as $(u)$.

(In the above, an IBKEM encapsulation and its related random value are respectively taken as the public-and-private parts of a hidden structure. To generate these two parts , an arbitrary keyword have to be chosen to run algorithm $\mathbf{Encaps}_{\mbox{\tiny IBKEM}}$.)

\item $\mathbf{StructuredEncryption}(\mathbf{PK},W,\mathbf{Pri})$: Take as inputs $\mathbf{PK}$, a keyword $W\in\mathcal{W}$, a hidden structure's private part $\mathbf{Pri}$, and do the following steps: \begin{enumerate}
\item Search $(W,Pt[u,W])$ for $W$ in $\mathbf{Pri}$;
\item If it is not found, insert $(W,Pt[u,W]\overset{\$}\leftarrow
\mathcal{M}_{\mbox{\tiny IBE}})$ to $\mathbf{Pri}$, and output the keyword-searchable ciphertext $C=(\mathbf{FIM}(W,u),\mathbf{Enc}_{\mbox{\tiny IBE}}(\mathbf{PK}_{\mbox{\tiny IBE}},W,Pt[u,W])$;
\item Otherwise, pick $R\overset{\$}{\leftarrow}\mathcal{M}_{\mbox{\tiny IBE}}$, set $C=(Pt[u,W],\mathbf{Enc}_{\mbox{\tiny IBE}}(\mathbf{PK}_{\mbox{\tiny IBE}},W,R))$, update $Pt[u,W]=R$, and output the keyword-searchable ciphertext $C$;
\end{enumerate}

\item $\mathbf{Trapdoor}(\mathbf{SK},W)$: Take as inputs $\mathbf{SK}$ and a keyword $W\in\mathcal{W}$, run $\hat{S}_W=\mathbf{Extract}_{\mbox{\tiny IBKEM}}($ $\mathbf{SK}_{\mbox{\tiny IBKEM}},W)$ and $\tilde{S}_W= \mathbf{Extract}_{\mbox{\tiny IBE}}(\mathbf{SK}_{\mbox{\tiny IBE}},W)$, and output a keyword search trapdoor $T_W=(\hat{S}_W,\tilde{S}_W)$ of keyword $W$.

\item $\mathbf{StructuredSearch}(\mathbf{PK},\mathbf{Pub},\mathbb{C},T_W)$: Take as inputs $\mathbf{PK}$, a hidden structure's public part $\mathbf{Pub}$, all keyword-searchable ciphertexts $\mathbb{C}$ (let $\mathbb{C}[i]$ denote one ciphertext of $\mathbb{C}$, and this ciphertext can be parsed as $(\mathbb{C}[i,1],\mathbb{C}[i,2])\in\mathcal{C}=\mathcal{K}_{\mbox{\tiny IBKEM}}\times\mathcal{C}_{\mbox{\tiny IBE}}$) and a keyword trapdoor $T_W=(\hat{S}_W,\tilde{S}_W)$ of keyword $W$, set $\mathbb{C}^\prime=\phi$, and do the following steps:
\begin{enumerate}
\item Compute $Pt^\prime=\mathbf{Decaps}_{\mbox{\tiny IBKEM}}(\hat{S}_W,\mathbf{Pub})$;
\item Seek a ciphertext $\mathbb{C}[i]$ having $\mathbb{C}[i,1]=Pt^\prime$; if 
it exists,  add $\mathbb{C}[i]$ into $\mathbb{C}^\prime$; 
\item If no matching ciphertext is found, output $\mathbb{C}^\prime$;
\item Compute $Pt^\prime=\mathbf{Dec}_{\mbox{\tiny IBE}}(\tilde{S}_{ID^\prime},\mathbb{C}[i,2])$, go to step 2;
\end{enumerate}
\end{itemize}

Figure \ref{F.SPCHS.Generic} shows a hidden star-like structure generated by the generic SPCHS construction. When running algorithm $\mathbf{StructuredSearch}(\mathbf{PK},\mathbf{Pub},\mathbb{C},T_{W_i})$, the full-identity malleability of IBKEM allows the algorithm to disclose the value $\mathbf{FIM}(W_i,u)$ by computing $\mathbf{FIM}(W_i,u)=\mathbf{Decaps}_{\mbox{\tiny IBKEM}}(\hat{S}_{W_i},\mathbf{Pub})$ and find out the ciphertext $(\mathbf{FIM}(W_i,u),\mathbf{Enc}_{\mbox{\tiny IBE}}(\mathbf{PK}_{\mbox{\tiny IBE}},$ $W_i,Pt[u,W_i]))$. Then the consistency of IBE allows the algorithm to disclose $Pt[u,W_i]$ by decrypting $\mathbf{Enc}_{\mbox{\tiny IBE}}(\mathbf{PK}_{\mbox{\tiny IBE}},$ $W_i,Pt[u,W_i])$ and find out the ciphertext $(Pt[u,W_i],\mathbf{Enc}_{\mbox{\tiny IBE}}($ $\mathbf{PK}_{\mbox{\tiny IBE}},$ $W_i,R))$. By carrying on in this way, 
the consistency of IBE allows the algorithm to find out the rest of ciphertexts of keyword $W_i$ with the hidden star-like structure, and stop the search if no more ciphertexts are found. 

\begin{figure*}[!htp]
\centering
\begin{boxedminipage}{0.9\textwidth}
\includegraphics[scale=0.7]{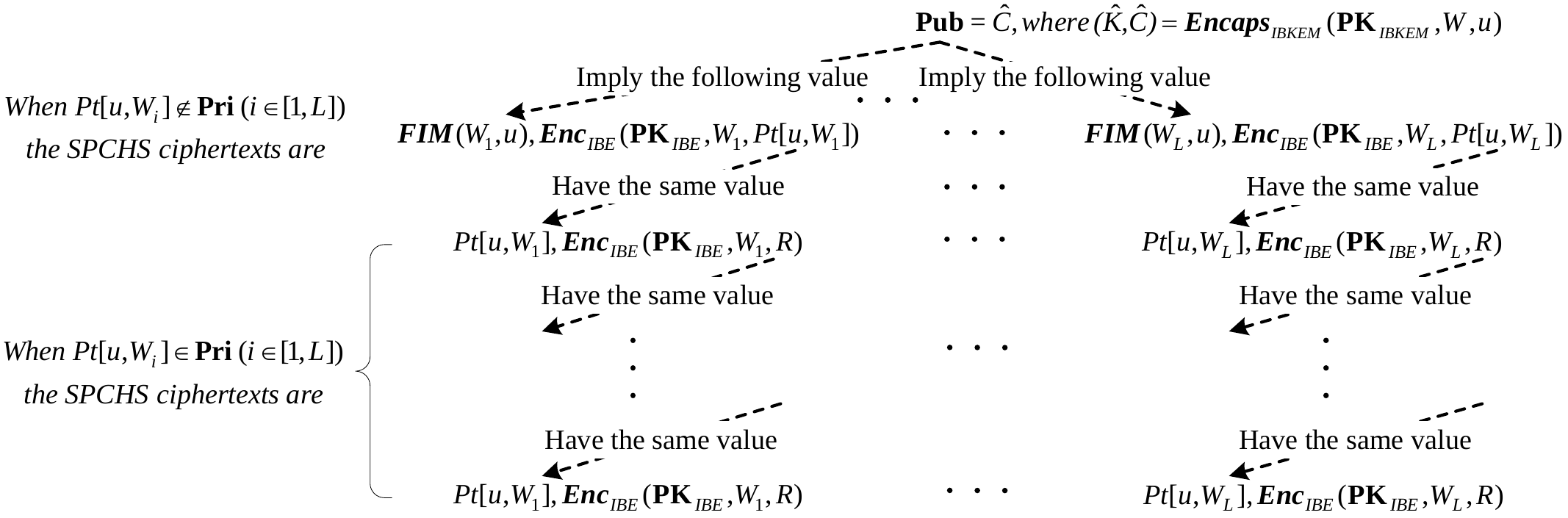}
\begin{quote}
Note that in each ciphertext, the value $R$ is randomly chosen. For $i\in[1,L]$, $Pt[u,W_i]$ is initialized with a random value when generating the first ciphertext of keyword $W_i$, and it will be updated into $R$ after generating each subsequent ciphertext of keyword $W_i$. 
\end{quote} 
\end{boxedminipage}
\caption{Hidden star-like structure generated by the generic SPCHS construction}\label{F.SPCHS.Generic}
\end{figure*}

\textbf{Consistency.} When running the above algorithm $\mathbf{StructuredSearch}(\mathbf{PK},\mathbf{Pub},\mathbb{C},T_W)$, the consistency and full-identity malleability of IBKEM assures that $\mathbf{FIM}(W,u)=\mathbf{Decaps}_{\mbox{\tiny IBKEM}}(\hat{S}_W,\mathbf{Pub})$ holds. The collision-freeness of IBKEM assures that only one ciphertext containing keyword $W$ has the value $\mathbf{FIM}(W,u)$ as its first part. Therefore the algorithm can find out the first ciphertext of keyword $W$ with the hidden structure $\mathbf{Pub}$. Then the consistency of IBE allows the algorithm $\mathbf{StructuredSearch}$ to find out the rest of ciphertexts containing keyword $W$ with the hidden structure $\mathbf{Pub}$. Formally we have Theorem \ref{T.SPCHS.Generic.Consistency}. The proof can be found in Supplemental Materials D.

\begin{theorem}\label{T.SPCHS.Generic.Consistency}
The above generic SPCHS scheme is consistent if its underlying collision-free full-identity malleable IBKEM and IBE schemes are both consistent.
\end{theorem}

\textbf{Semantic Security.} The SS-sK-CKSA security of the above generic SPCHS construction relies on the Anon-SS-sID-CPA security of the underlying IBKEM and the Anon-SS-ID-CPA security of the underlying IBE. In the security proof, we prove that if there is an adversary who can break the SS-sK-CKSA security of the above generic SPCHS construction, then there is another adversary who can break the Anon-SS-sID-CPA security of the underlying IBKEM or the Anon-SS-ID-CPA security of the underlying IBE. Theorem \ref{T.SPCHS.Generic.Security} formally states the semantic security of our generic SPCHS construction. The proof can be found in Supplemental Materials E.

\begin{theorem}\label{T.SPCHS.Generic.Security}
Suppose there are at most $N\in\mathbb{N}$ hidden structures, and a PPT adversary $\mathcal{A}$ wins in the SS-sK-CKSA game with advantage $Adv^{\text{SS-sK-CKSA}}_{SPCHS,\mathcal{A}}$. Then there is a PPT adversary $\mathcal{B}$, who utilizes the capability of $\mathcal{A}$ to win in the Anon-SS-sID-CPA game of the underlying IBKEM or the Anon-SS-ID-CPA game of the underlying IBE with advantage $\frac{1}{4N}\cdot Adv^{\text{SS-sK-CKSA}}_{SPCHS,\mathcal{A}}$.
\end{theorem}

\section{Two Collision-Free Full-Identity Malleable IBKEM Instances}\label{S.IBKEM.Instances}

\textbf{The Instance in the RO Model.} Abdalla \emph{et al.} proposed several VRF-suitable IBKEM instances in \cite{ACF13}. An IBKEM instance is VRF-suitable if it provides \emph{unique decapsulation}. This means that given any encapsulation, all the decryption keys corresponding to the same identity decapsulate out the same encapsulated key, and the key is pseudo-random. Here, the decryption key extraction is probabilistic and for the same identity, different decryption key may be extracted in different runs of the key extraction algorithm. It is clear that our proposed collision-free full-identity malleability not only implies \emph{unique decapsulation}, but also implies that the generator of an encapsulation knows what keys will be decapsulated by the decryption keys of all identities. In Supplemental Materials  F, we prove that the VRF-suitable IBKEM instance proposed in Appendix A.2 of \cite{ACF13} is collision-free full-identity malleable. Even though this IBKEM scheme has the traditional Anon-SS-ID-CPA security, we further prove that this IBKEM scheme is Anon-SS-ID-CPA secure based on the DBDH assumption in the RO model according to Definition \ref{D.IBKEM.Security}. 

\textbf{The Instance in the Standard Model.}  In \cite{FHPS13}, Freire \emph{et al.} utilized the ``approximation'' of multilinear maps \cite{GGH13} to construct a programmable hash function in the multilinear setting (MPHF). Then Freire \emph{et al.} utilized this hash function to replace the traditional hash functions of the BF IBE scheme in \cite{BF01} and reconstructed this IBE scheme in the multilinear setting. They finally constructed a new IBE scheme with semantic security in the standard model. We find that this new IBE scheme can be easily transformed into a collision-free full-identity malleable IBKEM scheme with Anon-SS-ID-CPA security in the standard model. To simplify the description of this IBKEM scheme, we do not consider the ``approximation'' of multilinear maps. This means that we will leave out the functions that are the encoding of a group element,  the re-randomization of an encoding and the extraction of an encoding.  Some related definitions are reviewed as follows.

\begin{definition}[Multilinear Maps \cite{FHPS13}]\label{D.Multilinear} 
An $\ell$-group system in the multilinear setting consists of $\ell$ cyclic groups $\mathbb{G}_1,\cdots ,\mathbb{G}_{\ell}$ of prime order $p$, along with bilinear maps $\hat{e}_{i,j} : \mathbb{G}_i \times \mathbb{G}_j \rightarrow \mathbb{G}_{i+j}$ for all $i,j \geq 1$ with $i+j \leq \ell$. Let $g_i$ be a generator of $\mathbb{G}_i$. The map $\hat{e}_{i,j}$ satisfies $\hat{e}_{i,j}(g^a_i,g^b_j) = g^{ab}_{i+j}$ (for all $a,b \in \mathbb{Z}_p$). When $i,j$ are clear, we will simply write $\hat{e}$ instead of $\hat{e}_{i,j}$. It will also be convenient to abbreviate $\hat{e}(h_1, \cdots, h_j) =\hat{e}(h_1,\hat{e}(h_2, \cdots, \hat{e}(h_{j-1},h_j)\cdots))$ for $h_j \in \mathbb{G}_{i_j}$ and $i = (i_1 + i_2 + \cdots + i_j) \leq \ell$. By induction, it is easy to see that this map is $j$-linear. Additionally, We define $\hat{e}(g)= g$. Finally, it can also be useful to define the group $\mathbb{G}_0 = \mathbb{Z}^+_{\left| \mathbb{G}_1 \right|}$ of exponents to which this pairing family naturally extends. In the following, we will assume an $\ell$-group system  $\mathbf{MPG}_{\ell} = \{\{\mathbb{G}_i\}_{i \in [1,\ell]},p, \{\hat{e}_{i,j}\}_{i,j\geq1,i+j\leq\ell}\}$ generated by a \emph{multilinear maps parameter generator} $\mathbf{MG}_{\ell}$ on input a security parameter $1^k$.
\end{definition}

\begin{definition}[The $\ell$-MDDH Assumption \cite{FHPS13}]\label{D.MDDH} 
Given $(g,g^{x_1},\cdots,$ $g^{x_{\ell+1}})$ (for $g\overset{\$}{\leftarrow} \mathbb{G}_1$ and uniform exponents $x_i$), the $\ell$-MDDH assumption is that the element $\hat{e}(g^{x_1},\cdots ,g^{x_\ell})^{x_{\ell+1}} \in \mathbb{G}_\ell$ is computationally indistinguishable from a uniform $\mathbb{G}_\ell$-element.
\end{definition}

\begin{definition}[Group hash function \cite{FHPS13}]\label{D.GroupHash}
A group hash function $\mathbf{H}$ into $\mathbb{G}$ consists of two polynomial-time algorithms: the probabilistic algorithm $\mathbf{HGen}(1^k)$ outputs a key $hk$, and $\mathbf{HEval}($ $hk,X)$ (for a key $hk$ and $X \in {\{0,1\}}^k$) deterministically outputs an image $\mathbf{H}_{hk}(X)\in\mathbb{G}$.
\end{definition}

\begin{definition}[MPHF \cite{FHPS13}]\label{D.MPHF}
Assume an $\ell^\prime$-group system $\mathbf{MPG}_{\ell^\prime}$ as generated by $\mathbf{MG}_{\ell^\prime}(1^k)$. Let $\mathbf{H}$ be a group hash function into $\mathbb{G}_{\ell}(\ell \leq \ell^\prime)$, and let $m,n \in \mathbb{N}$. We say that $\mathbf{H}$ is an (m,n)-programmable hash function in the multilinear setting ((m,n)-MPHF) if there are PPT algorithms $\mathbf{TGen}$ and $\mathbf{TEval}$ as follows.

\begin{itemize}
\item $\mathbf{TGen}(1^k,c_1,\cdots,c_l,h)$ (for $c_i,h \in \mathbb{G}_1$ and $h\neq 1$) outputs a key $hk$ and a trapdoor $td$. We require that for all $c_i,h$, that distribution of $hk$ is statistically close to the output of $\mathbf{HGen}$.

\item $\mathbf{TEval}(td,X)$ (for a trapdoor $td$ and $X \in {\{0,1\}}^k$) deterministically outputs $a_X \in \mathbb{Z}_p^*$ and $B_X \in \mathbb{G}_{\ell-1}$ with $\mathbf{H}_{hk}(X)= \hat{e}(c_1,\cdots,c_\ell)^{a_X} \cdot \hat{e}(B_X,h)$. We require that there is a polynomial $p(k)$ such that for all $hk$ and $X_1, \cdots, X_m$, $Z_1, \cdots, Z_n \in {\{0,1\}}^k$ with ${\{X_i\}}_i \bigcap {\{Z_j\}}_j = \emptyset$,
$P_{hk,\{X_i\},\{Z_j\}} =Pr[(a_{X_1} = \cdots = a_{X_m} = 0) \land (a_{Z_1},\cdots,a_{X_n}\neq0)]\geq 1/p(k)$,
where the probability is over possible trapdoors $td$ output by $\mathbf{TGen}$ along with the given $hk$. Furthermore, we require that 
    $P_{hk,\{X_i\},\{Z_j\}}$ is close to statistically independent of $hk$. (Formally, $|P_{hk,\{X_i\},\{Z_j\}}-P_{hk^\prime,\{X_i\},\{Z_j\}}|\leq v(k)$ for all $hk$ and $hk^\prime$ in the range of $\mathbf{TGen}$, 
    all $\{X_i\},\{Z_j\}$, and negligible $v(k)$.)
\end{itemize}
We say that $\mathbf{H}$ is a $(\mathsf{poly},n)$-MPHF if it is a $(q(k),n)$-MPHF for every polynomial $q(k)$. Note that $\mathbf{TEval}$ algorithm of an MPHF into $\mathbb{G}_1$ yields $B_X \in \mathbb{G}_0$, i.e., exponents $B_X$.
\end{definition}

Let identity space $\mathcal{ID}_{\mbox{\tiny IBKEM}}=\{0,1\}^k$. The IBKEM instance in the standard model is as follows.

\begin{itemize} 
\item $\mathbf{Setup}_{\mbox{\tiny IBKEM}}(1^k,\mathcal{ID}_{\mbox{\tiny IBKEM}})$: Take as input a security parameter $1^k$ and the identity space $\mathcal{ID}_{\mbox{\tiny IBKEM}}$, generate an $(\ell+1)$-group system $\mathbf{MPG}_{\ell+1}= \{\{\mathbb{G}_i\}_{i \in [1,\ell+1]}, p, \{\hat{e}_{i,j}\}_{i,j\geq1,i+j\leq\ell+1}\}\leftarrow\mathbf{MG}_{\ell+1}(1^k)$, generate a $(\mathsf{poly},2)$-MPHF $\mathbf{H}$ into $\mathbb{G}_\ell$ and $hk\leftarrow \mathbf{HGen}(1^k)$, choose $h\overset{\$}{\leftarrow} \mathbb{G}_1$ and  $x\overset{\$}{\leftarrow}\mathbb{Z}_p$, set the encapsulated key space $\mathcal{K}_{\mbox{\tiny IBKEM}}=\mathbb{G}_{\ell+1}$, set the encapsulation space $\mathcal{C}_{\mbox{\tiny IBKEM}}=\mathbb{G}_1$, and output the master public key $\mathbf{PK}_{\mbox{\tiny IBKEM}}=(\mathbf{MPG}_{\ell+1},hk,\mathbf{H},h,h^x,\mathcal{ID}_{\mbox{\tiny IBKEM}},\mathcal{K}_{\mbox{\tiny IBKEM}},$ $\mathcal{C}_{\mbox{\tiny IBKEM}})$ and the master secret key $\mathbf{SK}_{\mbox{\tiny IBKEM}}=(hk,x)$.

\item $\mathbf{Extract}_{\mbox{\tiny IBKEM}}(\mathbf{SK}_{\mbox{\tiny IBKEM}},ID)$: Take as inputs $\mathbf{SK}_{\mbox{\tiny IBKEM}}$ and an identity $ID\in\mathcal{ID}_{\mbox{\tiny IBKEM}}$, and output a decryption key $\hat{S}_{ID}=\mathbf{H}_{hk}(ID)^x$ of $ID$.

\item $\mathbf{Encaps}_{\mbox{\tiny IBKEM}}(\mathbf{PK}_{\mbox{\tiny IBKEM}},ID,r)$: Take as inputs $\mathbf{PK}_{\mbox{\tiny IBKEM}}$, an identity $ID\in\mathcal{ID}_{\mbox{\tiny IBKEM}}$ and a random value $r\in\mathbb{Z}_p^*$, and output a key-and-encapsulation pair $(\hat{K},\hat{C})$, where $\hat{K}=\hat{e}(\mathbf{H}_{hk}(ID),h^x)^r\in\mathbb{G}_{\ell+1}$ and $\hat{C}=h^r$.

\item $\mathbf{Decaps}_{\mbox{\tiny IBKEM}}(\hat{S}_{ID^\prime},\hat{C})$: Take as inputs the decryption key $\hat{S}_{ID^\prime}$ of identity $ID^\prime$ and an encapsulation $\hat{C}$, and output the encapsulated key $\hat{K}=\hat{e}(\hat{C},\hat{S}_{ID^\prime})\in\mathbb{G}_{\ell+1}$ if $\hat{C}\in\mathbb{G}_1$ or output $\bot$ otherwise.
\end{itemize}

\textbf{Consistency.} According to Definitions \ref{D.Multilinear} and \ref{D.GroupHash}, it is very easy to verify the consistency of the above IBKEM scheme. 

\textbf{Collision-Free Full-Identity Malleability.} Let the function $\mathbf{FIM}(ID,r)=\hat{e}(h^x,\mathbf{H}_{hk}(ID))^r\in\mathbb{G}_{\ell+1}$ for any identity $ID\in\mathcal{ID}_{\mbox{\tiny IBKEM}}$ and any random value $r\in\mathbb{Z}_p^*$. Given any $(\hat{K},\hat{C})\leftarrow\mathbf{Encaps}_{\mbox{\tiny IBKEM}}(\mathbf{PK}_{\mbox{\tiny IBKEM}},ID,r)$, we clearly have that: (1) for any identity $ID^\prime$, equation $\mathbf{FIM}(ID^\prime,$ $r)=\mathbf{Decaps}_{\mbox{\tiny IBKEM}}(\hat{S}_{ID^\prime},\hat{C})$ holds; (2) for any identity $ID^\prime$ and any random value $r^\prime$, if $ID^\prime\neq ID\bigvee r^\prime\neq r$ holds, equation $\mathbf{FIM}(ID,r)\neq \mathbf{FIM}(ID^\prime,r^\prime)$ holds except with a negligible probability. So the above IBKEM scheme is collision-free full-identity malleable.

\textbf{Anon-SS-ID-CPA Security.} In \cite{FHPS13}, Freire \emph{et al.} utilized a $(\mathsf{poly},1)$-MPHF to construct a standard-model version of the BF IBE scheme with the SS-ID-CPA security. On the contrary, we use a $(\mathsf{poly},2)$-MPHF in constructing the above IBKEM scheme, since this kind of MPHF is more useful in proving the Anon-SS-ID-CPA security. Theorem \ref{T.IBKEM.Instance.Security} formally states the Anon-SS-ID-CPA security of the above IBKEM scheme. The proof can be found in Supplemental Materials  G.

\begin{theorem}\label{T.IBKEM.Instance.Security}
Assume the above IBKEM scheme is implemented in an $(\ell+1)$-group system, and with a $(\mathsf{poly},2)$-MPHF $\mathbf{H}$ into $\mathbb{G}_\ell$. Then, under the $(\ell+1)$-MDDH assumption, this IBKEM scheme is Anon-SS-ID-CPA secure.
\end{theorem}

According to Theorem \ref{T.SPCHS.Generic.Security} and \ref{T.IBKEM.Instance.Security}, the generic SPCHS construction implies a SPCHS instance with SS-sK-CKSA security in the standard model. Indeed, this SPCHS instance can be provably SS-CKSA secure. 

\section{Conclusion and Future Work}\label{S.Conclusion}

This paper investigated as-fast-as-possible search in PEKS with semantic security. We proposed the concept of SPCHS as a variant of PEKS. The new concept allows keyword-searchable ciphertexts to be generated with a hidden structure. Given a keyword search trapdoor, the search algorithm of SPCHS can disclose part of this hidden structure for guidance on finding out the ciphertexts of the queried keyword. Semantic security of SPCHS captures the privacy of the keywords and the invisibility of the hidden structures. We proposed an SPCHS scheme from scratch with semantic security in the RO model. The scheme generates keyword-searchable ciphertexts with a hidden star-like structure. It has search complexity mainly linear with the exact number of the ciphertexts containing the queried keyword. It outperforms existing PEKS schemes with semantic security, whose search complexity is linear with the number of all ciphertexts. We identified several interesting properties, i.e., collision-freeness and full-identity malleability in some IBKEM instances, and formalized these properties to build a generic SPCHS construction. We illustrated two collision-free full-identity malleable IBKEM instances, which are respectively secure in the RO and standard models. 

SPCHS seems a promising tool to solve some challenging problems in public-key searchable encryption. One application may be to achieve retrieval completeness verification which, to the best of our knowledge, has not been achieved in existing PEKS schemes. Specifically, by forming a hidden ring-like structure, i.e., letting the last hidden pointer always point to the head, one can obtain PEKS allowing to check the completeness of the retrieved ciphertexts by checking whether the pointers of the returned ciphertexts form a ring. 

Another application may be to realize public key encryption with content search, a similar functionality realized by symmetric searchable encryption. Such kind  of content-searchable encryption is useful in practice, e.g., to filter the encrypted spams. Specially, by forming a hidden tree-like structure between the sequentially encrypted words in one file, one can obtain public-key searchable encryption allowing content search (e.g., to find whether there are specific contents in an encrypted file). The search complexity is linear with the size of the queried content.

\section*{Acknowledgments}
The authors would like to thank the reviewers for their valuable suggestions that helped to improve the paper greatly. The first author is partly supported by the National Natural Science Foundation of China under grant no. 61472156 and the National Program on Key Basic Research Project (973 Program) under grant no. 2014CB340600. The second author is supported by by the Chinese National Key Basic Research Program (973 program) through project 2012CB315905, the Natural Science Foundation of China through projects 61370190, 61173154, 61472429, 61402029, 61272501, 61202465, 61321064 and 61003214, the Beijing Natural Science Foundation through project 4132056, the Fundamental Research Funds for the Central Universities, and the Research Funds (No. 14XNLF02) of Renmin University of China and the Open Research Fund of Beijing Key Laboratory of Trusted Computing. The fifth author is partly support by the European Commission (H2020 project ``CLARUS''), the Government of Catalonia (grant 2014 SGR 537 and ICREA Acad\`emia Award 2013) and the Spanish Government (TIN2011-27076-C03-01 ``CO-PRIVACY''). The views in this paper do not necessarily reflect the views
of UNESCO.

\begin{biography}[{\includegraphics[width=1in,height=1.25in,clip,keepaspectratio]{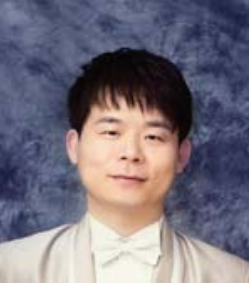}}]{Peng Xu}
received the B.A. degree in computer science from Wuhan university
of science and technique, Wuhan, China, in 2003, the Master and
Ph.D. degree in computer science from Huazhong university of science
and technology, Wuhan, China, respectively in 2006 and 2010. Since
2010, he works as a post-doctor at Huazhong university of science
and technology, Wuhan, China. He was PI in three grants respectively from National Natural Science Foundation of China (No. 61472156 and No. 61100222) and China Postdoctoral Science Foundation (No. 20100480900), and a key member in several projects supported by 973 (No. 2014CB340600). He has authored over 20 research papers. He is a member of ACM and IEEE.
\end{biography}

\begin{biography}[{\includegraphics[width=1in,height=1.25in,clip,keepaspectratio]{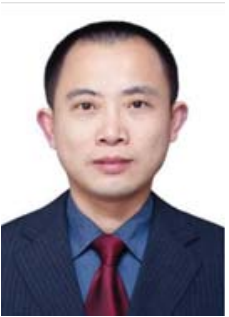}}]{Qianhong Wu}
received his Ph.D. in Cryptography from Xidian University in 2004. Since then, he has been with Wollongong University (Australia) as an associate research fellow, with Wuhan University (China) as an associate professor, with Universitat Rovira i Virgili (Catalonia) as a research director and now with Beihang University (China) as a full professor. His research interests include cryptography, information security and privacy, and \emph{ad hoc} network security. He has been a holder/co-holder of 7 China/Australia/Spain funded projects. He has authored 7 patents and over 100 publications. He has served in the program committee of several international conferences in information security and privacy. He is a member of IACR, ACM and IEEE.
\end{biography}

\begin{biography}[{\includegraphics[width=1in,height=1.25in,clip,keepaspectratio]{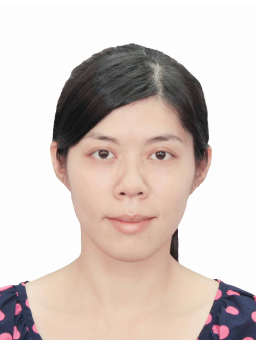}}]{Wei Wang}
received the B.S. and Ph.D. degrees in Electronic and Communication Engineering from Huazhong University of Science and Technology, Wuhan, China, in 2006 and 2011, respectively.
Currently she is a researcher with Cyber-Physical-Social Systems Lab, Huazhong University of Science and Technology, Wuhan, China. She was a Post-doctoral researcher with Peking University, Beijing, China from April 2012 to July 2014. Her research interests include cloud security, network coding and multimedia transmission. She has published over 10 papers in international journals and conferences.
\end{biography}

\begin{biography}[{\includegraphics[width=1in,height=1.25in,clip,keepaspectratio]{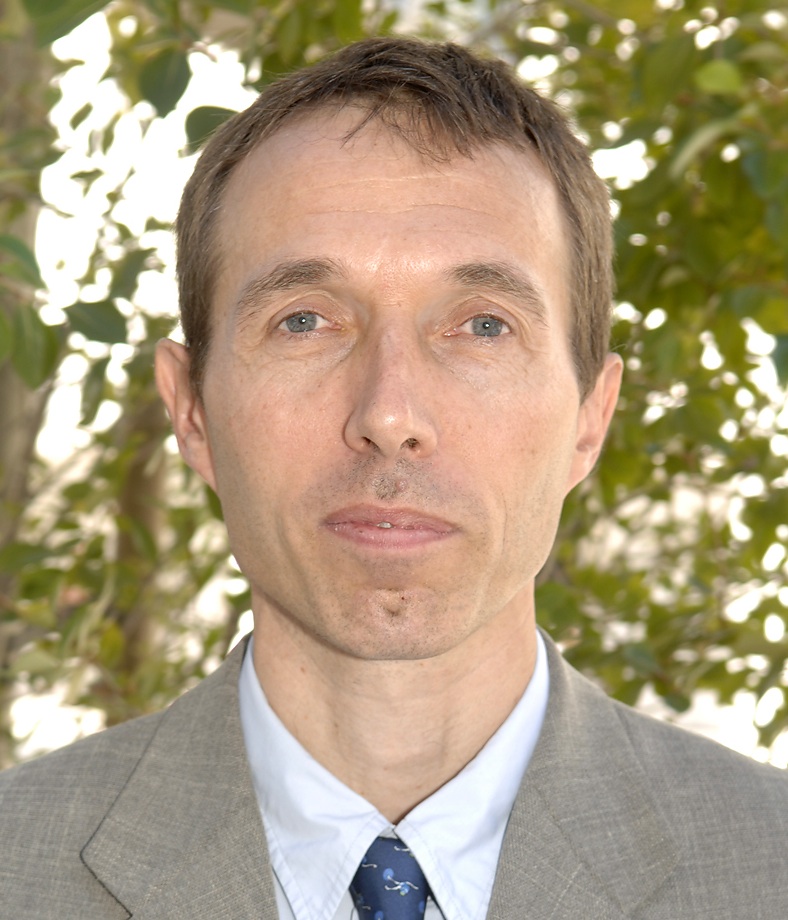}}]{Josep Domingo-Ferrer} is a Distinguished Professor of Computer Science and an ICREA-Acad\`emia Researcher at Universitat Rovira i Virgili, Tarragona, Catalonia, where he holds the UNESCO Chair in Data Privacy. His research interests are in data privacy and data security. He received his M. Sc. and Ph. D. degrees in Computer Science from the Autonomous University of Barcelona in 1988 and 1991, respectively. He also holds an M. Sc. in Mathematics. He has won several research and technology transfer awards, including twice the ICREA Academia Prize (2008 and 2013) and the ``Narc\'{\i}s Monturiol'' Medal to the Scientific Merit, both awarded by the Government of Catalonia, and a Google Faculty Research Award (2014). He has authored 5 patents and over 350 publications. He has been the co-ordinator of projects funded by the European Union and the Spanish government.
He has been the PI of US-funded research contracts and currently of a Templeton World Charity Foundation grant. He has held visiting appointments at Princeton, Leuven and Rome. He is a co-Editor-in-Chief of {\em Transactions on Data Privacy}. He is an IEEE Fellow and an Elected Member of Academia Europaea.
\end{biography}

\begin{biography}[{\includegraphics[width=1in,height=1.25in,clip,keepaspectratio]{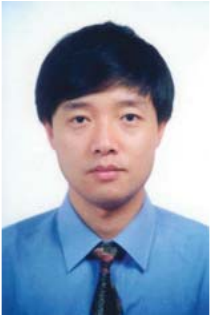}}]{Hai Jin}
received his PhD in computer engineering from HUST in 1994. In 1996,
he was awarded a German Academic Exchange Service fellowship to
visit the Technical University of Chemnitz in Germany. He worked at
The University of Hong Kong between 1998 and 2000, and as a visiting
scholar at the University of Southern California between 1999 and
2000. He was awarded Excellent Youth Award from the National Science
Foundation of China in 2001. He is the chief scientist
of National 973 Basic Research Program Project of Virtualization
Technology of Computing System. He has co-authored 15 books and
published over 400 research papers. He is a senior member of the IEEE and
a member of the ACM.
\end{biography}

\newpage

\begin{onecolumn}

\pagestyle{empty}

\section*{Supplemental Materials}

\subsection{Analysis on The Work \cite{CKRS09}}\label{Analysis on Insecurity}

\begin{figure*}[!htp]
\centering
\begin{boxedminipage}{0.9\textwidth}
A sender generates the searchable ciphertexts of any keyword $W_i\in\mathcal{W}$ by the following steps:
\begin{enumerate}
\item The first time to encrypt keyword $W_i$, he uploads $$PEKS(Pub,W_i,K_i^1||P_i^1),P_i^1||E(K_i^1,P_i^2||K_i^2||\mathcal{P}_i^1),P_i^2$$ to the server, and asks the server to store $E(K_i^1,P_i^2||K_i^2||\mathcal{P}_i^1)$ in position $P_i^1$ and store a flag in position $P_i^2$.

Note: algorithm $PEKS(Pub,W_i^1,K_i^1||P_i^1)=IBE(Pub,W_i^1,K_i^1||P_i^1||C_2)||C_2$ takes public parameter $Pub$, identity $W_i^1$ and plaintext $K_i^1||P_i^1||C_2$ as inputs and generates an IBE ciphertext, and finally outputs the IBE ciphertext and $C_2$, where the symmetric key $K_i^1$ and $C_2$ are randomly chosen. Algorithm $E(K_i^1,P_i^2||K_i^2||\mathcal{P}_i^1)$ denotes using the symmetric key $K_i^1$ to encrypt $P_i^2||K_i^2||\mathcal{P}_i^1$, where the symmetric key $K_i^2$ is randomly chosen, and  $\mathcal{P}_i^1$ denotes the parameters for private information retrieval (they will be used to retrieve the corresponding data when the keyword $W_i$ is queried).

\item The second time to encrypt keyword  $W_i$, he uploads $P_i^2||E(K_i^2,P_i^3||K_i^3||\mathcal{P}_i^2),P_i^3$ to the server, and asks the server to store $E(K_i^2,P_i^3||K_i^3||\mathcal{P}_i^2)$ in position $P_i^2$ and store the flag in position $P_i^3$ .

\item The subsequent encryptions of keyword $W_i$ are similar to Step 2.
\end{enumerate}
\end{boxedminipage}
\caption{Procedure to generate keyword searchable ciphertexts in \cite{CKRS09}.}\label{Insecure Chain}
\end{figure*}

In Fig. \ref{Insecure Chain}, we first review how to generate keyword-searchable ciphertexts according to \cite{CKRS09} such that the ciphertexts of the same keyword form a chain. Then we analyze why the chain of any keyword is visible in the view of the server, and give a straightforward method to make the chain invisible. But this method seems to be impractical.

According to the first step in Fig. \ref{Insecure Chain}, the server trivially knows the relation between ciphertexts $PEKS(Pub,W_i,K_i^1||P_i^1)$ and $E(K_i^1,P_i^2||K_i^2||\mathcal{P}_i^1)$, and knows that if a subsequent ciphertext is stored in the position $P_2$, this subsequent ciphertext is related to $E(K_i^1,P_i^2||K_i^2||\mathcal{P}_i^1)$. So in the second step, the server knows the relation between ciphertexts $E(K_i^1,P_i^2||K_i^2||\mathcal{P}_i^1)$ and $E(K_i^2,P_i^3||K_i^3||\mathcal{P}_i^2)$, and knows that if another subsequent ciphertext is stored in the position $P_3$, this subsequent ciphertext is related to $E(K_i^2,P_i^3||K_i^3||\mathcal{P}_i^2)$. By the same method, the server will know the chain of keyword $W_i$ even without the keyword search trapdoor of keyword $W_i$. Furthermore, the length of the chain leaks the frequency of keyword $W_i$.

\begin{figure*}[!htp]
\centering
\begin{boxedminipage}{0.9\textwidth}
A sender generates the searchable ciphertexts of any keyword $W_i\in\mathcal{W}$ by the following steps:
\begin{enumerate}
\item At the setup phase, he uploads $\{PEKS(Pub,W_i,K_i^1||P_i^1)|i\in[1,|\mathcal{W}|]\}$ to the server, where $|\mathcal{W}|$ denotes the size of keyword space $\mathcal{W}$.

\item The first time to encrypt keyword $W_i$, he uploads $P_i^1||E(K_i^1,P_i^2||K_i^2||\mathcal{P}_i^1)$ to the server, and asks the server to store $E(K_i^1,P_i^2||K_i^2||\mathcal{P}_i^1)$ in position $P_i^1$.

\item The second time to encrypt keyword  $W_i$, he uploads $P_i^2||E(K_i^2,P_i^3||K_i^3||\mathcal{P}_i^2)$ to the server, and asks the server to store $E(K_i^2,P_i^3||K_i^3||\mathcal{P}_i^2)$ in position $P_i^2$.

\item The subsequent encryptions of keyword $W_i$ are similar to Step 3.
\end{enumerate}
\end{boxedminipage}
\caption{New procedure to generate keyword-searchable ciphertexts for \cite{CKRS09}.}\label{Secure Chain}
\end{figure*}

In order to keep the privacy of the chain, a straightforward method is to generate the PEKS ciphertexts for all keywords at the setup phase and delete the flag. The specific procedure is given in Fig. \ref{Secure Chain}. This method hides the relation between the PEKS ciphertext and the symmetric-key ciphertext of any keyword, and the relation between two symmetric-key ciphertexts of any keyword also is hidden. But it seems that this method is impractical from 
a performance viewpoint, since each sender must generate the PEKS ciphertexts for all keywords at the setup phase and remember lots of private information which are encrypted by these PEKS ciphertexts. 

\subsection{Proof of Theorem \ref{T.SPCHS.Instance.Consistency}}\label{P.SPCHS.Instance.Consistency}

\begin{proof}
Without loss of generality, it is sufficient to prove that given the keyword-searchable trapdoor $T_{W_i}=H(W_i)^s$ of keyword $W_i$ and the hidden structure's public part $\mathbf{Pub}=g^u$, algorithm $\mathbf{StructuredSearch}(\mathbf{PK},\mathbf{Pub},\mathbb{C},T_{W_i})$ only finds out all ciphertexts of keyword $W_i$ with the hidden structure $\mathbf{Pub}$. Note that $P=g^s$. 

Algorithm $\mathbf{StructuredSearch}(\mathbf{PK},\mathbf{Pub},\mathbb{C},T_{W_i})$  computes $Pt^\prime=\hat{e}(\mathbf{Pub},T_{W_i})$ in its first step. Since $\hat{e}(\mathbf{Pub},T_{W_i})=\hat{e}(P,H(W_i))^u$, algorithm $\mathbf{StructuredSearch}(\mathbf{PK},\mathbf{Pub},\mathbb{C},T_{W_i})$ finds out the ciphertext $(\hat{e}(P,H(W_i))^u,g^r,\hat{e}(P,H(W_i))^r\cdot Pt[u,W_i])$ by matching $Pt^\prime$ with all ciphertexts' first part in its second step. Moreover, due to the collision-freeness of hash function $H$, only keyword $W_i$ has $Pt^\prime=\hat{e}(P,H(W_i))^u$, except with a negligible probability in the security parameter $k$. So only the ciphertext $(\hat{e}(P,H(W_i))^u,g^r,\hat{e}(P,H(W_i))^r\cdot Pt[u,W_i])$ is found with overwhelming probability in this step.

Then algorithm $\mathbf{StructuredSearch}(\mathbf{PK},\mathbf{Pub},\mathbb{C},T_{W_i})$ discloses $Pt[u,W_i]$ from the ciphertext $(\hat{e}(P,H(W_i))^u,g^r,\hat{e}(P,H(W_i))^r\cdot Pt[u,W_i])$ by computing $Pt^\prime=Pt[u,W_i]=\hat{e}(g^r,T_{W_i})^{-1}\cdot\hat{e}(P,H(W_i))^r\cdot Pt[u,W_i]$.

Recall that in algorithm $\mathbf{StructuredEncryption}$, $Pt[u,W_i]$ was randomly chosen in $\mathbb{G}_1$ and taken as the first part of only one ciphertext of keyword $W_i$ with the hidden structure $\mathbf{Pub}$. So when algorithm $\mathbf{StructuredSearch}(\mathbf{PK},\mathbf{Pub},\mathbb{C},T_{W_i})$ goes back to its second step, only the ciphertext $(Pt[u,W_i],g^r,\hat{e}(P,H(W_i))^r\cdot R)$ is found with overwhelming probability. 

By carrying on in this way, algorithm $\mathbf{StructuredSearch}(\mathbf{PK},$ $\mathbf{Pub},\mathbb{C},T_{W_i})$ only finds out all ciphertexts of keyword $W_i$ with the hidden structure $\mathbf{Pub}$, except with a negligible probability in the security parameter $k$ . And the algorithm will stop, since the random value $R$ contained in the last found ciphertext does not match any other ciphertext's first part.
\end{proof}

\subsection{Proof of Theorem \ref{T.SPCHS.Instance.Security}}\label{P.SPCHS.Instance.Securtiy}

\begin{proof}
To prove this theorem, we will construct a PPT algorithm $\mathcal{B}$ that plays the SS-CKSA game with adversary $\mathcal{A}$ and utilizes the capability of $\mathcal{A}$ to solve the DBDH problem in $\mathbf{BGen}(1^k)$ with advantage approximately
$\frac{27}{(e\cdot q_t\cdot q_p)^3}\cdot Adv^{\text{SS-CKSA}}_{SPCHS,\mathcal{A}}$. Let $Coin\overset{\sigma}\leftarrow \{0,1\}$ denote the operation that picks $Coin\in\{0,1\}$ according to the probability $Pr[Coin = 1] = \sigma$ (the specified value of $\sigma$ will be decided latter). The constructed algorithm $\mathcal{B}$ in the SS-CKSA game is as follows.

\begin{itemize}
\item \textbf{Setup Phase}: Algorithm $\mathcal{B}$ takes as inputs $(q,\mathbb{G},\mathbb{G}_1,g,\hat{e},g^a,g^b,g^c,Z)$ (where $Z$ equals either $\hat{e}(g,g)^{abc}$ or $\hat{e}(g,g)^y$) and the keyword space $\mathcal{W}$, and 
performs the following steps:
\begin{enumerate}
\item Initialize the three lists $\mathbf{Pt}=\emptyset\subseteq\mathcal{W}\times\mathbb{G}\times\mathbb{G}_1$, $\mathbf{SList}=\emptyset\subseteq\mathbb{G}\times\mathbb{Z}_q^*\times\{0,1\}$ and $\mathbf{HList}=\emptyset\subseteq\mathcal{W}\times\mathbb{G}\times\mathbb{Z}_q^*\times \{0,1\}$;
\item Set the ciphertext space $\mathcal{C}=\mathbb{G}_1\times\mathbb{G}\times\mathbb{G}_1$ and $\mathbf{PK}=(q,\mathbb{G},\mathbb{G}_1,g,\hat{e},P=g^a,\mathcal{W},\mathcal{C})$;
\item Initialize $N$ hidden structures by repeating the following steps for $i\in[1,N]$: 
\begin{enumerate}
\item Pick $u_i\overset{\$}\leftarrow\mathbb{Z}_q^*$ and $Coin_i\overset{\sigma}\leftarrow\{0,1\}$;
\item If $Coin_i=1$, compute $\mathbf{Pub}_i=g^{b\cdot u_i}$; 
\item Otherwise, compute $\mathbf{Pub}_i=g^{u_i}$;
\end{enumerate}
\item Set $\mathbf{PSet}=\{\mathbf{Pub}_i|i\in[1,N]\}$ and $\mathbf{SList}=\{(\mathbf{Pub}_i,u_i,Coin_i)|i\in[1,N]\}$;
\item Send $\mathbf{PK}$ and $\mathbf{PSet}$ to adversary $\mathcal{A}$.
\end{enumerate}

\item \textbf{Query Phase 1}: Adversary $\mathcal{A}$ adaptively issues the following queries multiple times.
\begin{itemize}
\item Hash Query $\mathcal{Q}_{H}(W)$: Taking as input a keyword $W\in\mathcal{W}$, algorithm $\mathcal{B}$ does the following steps:
	\begin{enumerate}
	\item Pick $x\overset{\$}\leftarrow\mathbb{Z}_q^*$ and $Coin\overset{\sigma}\leftarrow\{0,1\}$;
	\item If $Coin=0$, add $(W,g^x,x,Coin)$ into $\mathbf{HList}$ and output $g^x$;
	\item Otherwise, add $(W,g^{c\cdot x},x,Coin)$ into $\mathbf{HList}$ and output $g^{c\cdot x}$;
	\end{enumerate}
\item Trapdoor Query $\mathcal{Q}_{Trap}(W)$: Taking as input a keyword $W\in\mathcal{W}$, algorithm $\mathcal{B}$ does the following steps:
	\begin{enumerate}
	\item If $(W,*,*,*)\notin\mathbf{HList}$, query $\mathcal{Q}_H(W)$;
	\item According to $W$, retrieve $(W,X,x,Coin)$ from $\mathbf{HList}$;
	\item If $Coin=0$, output $g^{a\cdot x}$; otherwise abort and output $\bot$;
	\end{enumerate}
\item Privacy Query $\mathcal{Q}_{Pri}(\mathbf{Pub})$: Taking as input a structure's public part $\mathbf{Pub}\in\mathbf{PSet}$, algorithm $\mathcal{B}$ does the following steps: 
	\begin{enumerate}
	\item According to $\mathbf{Pub}$, retrieve $(\mathbf{Pub},u,Coin)$ from $\mathbf{SList}$;
	\item If $Coin=0$, output $u$; otherwise abort and output $\bot$;
	\end{enumerate}
\item Encryption Query $\mathcal{Q}_{Enc}(W,\mathbf{Pub})$: Taking as inputs a keyword $W\in\mathcal{W}$ and a structure's public part $\mathbf{Pub}$, algorithm $\mathcal{B}$ does the following steps:
	\begin{enumerate}
	\item If $(W,*,*,*)\notin\mathbf{HList}$, query $\mathcal{Q}_H(W)$;
	\item According to $W$ and $\mathbf{Pub}$, retrieve $(W,X,x,Coin)$ and $(\mathbf{Pub},u,Coin^\prime)$ respectively from $\mathbf{HList}$ and $\mathbf{SList}$; 
	\item Pick $r\overset{\$}\leftarrow\mathbb{Z}_q^*$, and search $(W,\mathbf{Pub},Pt[u,W])$ for $W$ and $\mathbf{Pub}$ in $\mathbf{Pt}$; 
	\item If $W$ is not found, insert $(W,\mathbf{Pub},Pt[u,W]\overset{\$}\leftarrow\mathbb{G}_1)$ to $\mathbf{Pt}$ and do the following steps:
	\begin{enumerate}
	\item If $Coin=1\bigwedge Coin^\prime=1$, output $C=(Z^{x\cdot u},g^r,\hat{e}(g^a,X)^r\cdot Pt[u,W])$;
	\item If $Coin=0\bigwedge Coin^\prime=1$, output $C=(\hat{e}(g^a,g^{b\cdot u})^x,g^r,\hat{e}(g^a,X)^r\cdot Pt[u,W])$;
	\item If $Coin^\prime=0$, output $C=(\hat{e}(g^a,X)^u,g^r,\hat{e}(g^a,X)^r\cdot Pt[u,W])$;
	\end{enumerate}
	\item Otherwise, pick $R\overset{\$}{\leftarrow}\mathbb{G}_1$, 
set $C=(Pt[u,W],g^r,\hat{e}(g^a,X)^r\cdot R)$, update $Pt[u,W]=R$ and output $C$;
	\end{enumerate}
 \end{itemize}

\item \textbf{Challenge Phase}: Adversary $\mathcal{A}$ sends two challenge keyword-structure pairs $(W^*_0,\mathbf{Pub}^*_0)\in\mathcal{W}\times\mathbf{PSet}$ and $(W^*_1,\mathbf{Pub}^*_1)\in\mathcal{W}\times\mathbf{PSet}$ to algorithm $\mathcal{B}$; $\mathcal{B}$ picks $d\overset{\$}\leftarrow\{0,1\}$, and does the following steps:
\begin{enumerate}
\item  According to $\mathbf{Pub}_0^*$ and $\mathbf{Pub}_1^*$, retrieve $(\mathbf{Pub}_0^*,u_0^*,Coin_0^*)$ and $(\mathbf{Pub}_1^*,u_1^*,Coin_1^*)$ from $\mathbf{SList}$; and if $Coin_0^*=0\bigvee Coin_1^*=0$, then abort and output $\bot$;
\item If $(W_d^*,*,*,*)\notin\mathbf{HList}$, query $\mathcal{Q}_{H}(W_d^*)$;
\item According to $W_d^*$, retrieve $(W_d^*,X_d^*,x_d^*,Coin)$ from $\mathbf{HList}$; and if $Coin=0$, then abort and output $\bot$;
\item Search $(W_d^*,\mathbf{Pub}_d^*,Pt[u_d^*,W_d^*])$ for $W_d^*$ and $\mathbf{Pub}_d^*$ in $\mathbf{Pt}$; 
\item If it is not found, insert $(W_d^*,\mathbf{Pub}_d^*,Pt[u_d^*,W_d^*]\overset{\$}\leftarrow
\mathbb{G}_1)$ to $\mathbf{Pt}$, and send $C^*_d=(Z^{x_d^*\cdot u_d^*},g^b,Z^{x_d^*}\cdot Pt[u_d^*,W_d^*])$ to adversary $\mathcal{A}$; 
\item Otherwise, pick $R\overset{\$}{\leftarrow}\mathbb{G}_1$, 
set $C_d^*=(Pt[u_d^*,W_d^*],g^b,Z^{x_d^*}\cdot R)$, update $Pt[u_d^*,W_d^*]=R$, and send $C_d^*$ to adversary $\mathcal{A}$;
\end{enumerate}

\item \textbf{Query Phase 2}: This phase is the same as \textbf{Query Phase 1}. Note that in \textbf{Query Phase 1} and \textbf{Query Phase 2}, adversary $\mathcal{A}$ cannot query the corresponding private parts both of $\mathbf{Pub}^*_0$ and $\mathbf{Pub}^*_1$ and the keyword search trapdoors both of $W^*_0$ and $W^*_1$.

\item \textbf{Guess Phase}: Adversary $\mathcal{A}$ sends a guess $d^\prime$ to algorithm $\mathcal{B}$. If $d=d^\prime$, $\mathcal{B}$ output 1; otherwise, output 0.
\end{itemize}

Let $\overline{Abort}$ denote the event that algorithm $\mathcal{B}$ does not abort in the above game. Next, we will compute the probabilities $Pr[\overline{Abort}]$, $Pr[\mathcal{B}=1|Z=\hat{e}(g,g)^{abc}]$ and $Pr[\mathcal{B}=1|Z=\hat{e}(g,g)^y]$, and the advantage $Adv_\mathcal {B}^{DBDH}(1^k)$.

According to the above game, the probability of the event $\overline{Abort}$ only relies on the probability $\sigma$ and the number of times that adversary $\mathcal{A}$ queries oracles $\mathcal{Q}_{Trap}(\cdot)$ and $\mathcal{Q}_{Pri}(\cdot)$. We have that $Pr[\overline{Abort}]=(1-\sigma)^{q_t\cdot q_p}\cdot \sigma^3$. Let $\sigma=\frac{3}{3+q_t\cdot q_p}$. We have that  $Pr[\overline{Abort}]\approx\frac{27}{(e\cdot q_t\cdot q_p)^3}$, where $e$ is the base of natural logarithms.

When $Z=\hat{e}(g,g)^{abc}$ and the event $\overline{Abort}$ holds, it is easy to find that algorithm $\mathcal{B}$ simulates a real SS-CKSA game in adversary $\mathcal{A}$'s mind. So we have $$Pr[d=d^\prime|\overline{Abort}\bigwedge Z=\hat{e}(g,g)^{abc}]=(Adv^{\text{SS-CKSA}}_{SPCHS,\mathcal{A}}+\frac{1}{2}).$$
 
When $Z=\hat{e}(g,g)^y$ and the event $\overline{Abort}$ holds, algorithm $\mathcal{B}$ generates a challenge ciphertext, which is independent of the challenge keywords $W_0^*$ and $W_1^*$. So we have $$Pr[d=d^\prime|\overline{Abort}\bigwedge Z=\hat{e}(g,g)^y]=\frac{1}{2}.$$

Now, we can compute the advantage $Adv_\mathcal{B}^{DBDH}(1^k)$ as follows:
\begin{equation}
\begin{aligned}
Adv_\mathcal{B}^{DBDH}(1^k)&=Pr[\mathcal{B}=1|Z=\hat{e}(g,g)^{abc}]-Pr[\mathcal{B}=1|Z=\hat{e}(g,g)^y]\\
&=Pr[d=d^\prime\bigwedge\overline{Abort}|Z=\hat{e}(g,g)^{abc}]-Pr[d=d^\prime\bigwedge\overline{Abort}|Z=\hat{e}(g,g)^y]\\
&=Pr[d=d^\prime|\overline{Abort}\bigwedge Z=\hat{e}(g,g)^{abc}]\cdot Pr[\overline{Abort}|Z=\hat{e}(g,g)^{abc}]\\&\qquad -Pr[d=d^\prime|\overline{Abort}\bigwedge Z=\hat{e}(g,g)^y]\cdot Pr[\overline{Abort}|Z=\hat{e}(g,g)^y]\\
&\approx(Adv^{\text{SS-CKSA}}_{SPCHS,\mathcal{A}}+\frac{1}{2})\cdot\frac{27}{(e\cdot q_t\cdot q_p)^3}-\frac{1}{2}\cdot\frac{27}{(e\cdot q_t\cdot q_p)^3}\\& \approx\frac{27}{(e\cdot q_t\cdot q_p)^3}\cdot Adv^{\text{SS-CKSA}}_{SPCHS,\mathcal{A}} \nonumber
\end{aligned}
\end{equation}

In addition, it is clear that algorithm $\mathcal{B}$ is a PPT algorithm, if adversary $\mathcal{A}$ is a PPT adversary. In conclusion, if a PPT adversary $\mathcal{A}$ wins in the SS-CKSA game of the above SPCHS instance with advantage $Adv^{\text{SS-CKSA}}_{SPCHS,\mathcal{A}}$, in which
$\mathcal{A}$ makes at most $q_t$ queries to oracle $\mathcal{Q}_{Trap}(\cdot)$ and at most $q_p$ queries to oracle $\mathcal{Q}_{Pri}(\cdot)$, then there is a PPT algorithm $\mathcal{B}$ that solves the
DBDH problem in $\mathbf{BGen}(1^k)$ with advantage approximately
$$Adv_\mathcal {B}^{DBDH}(1^k)\approx\frac{27}{(e\cdot q_t\cdot q_p)^3}\cdot Adv^{\text{SS-CKSA}}_{SPCHS,\mathcal{A}}$$ where $e$ is the base of natural logarithms.
\end{proof}

\subsection{Proof of Theorem \ref{T.SPCHS.Generic.Consistency}}
\label{P.SPCHS.Generic.Consistency}

\begin{proof}
Without loss of generality, it is sufficient to prove that given the keyword-searchable trapdoor $T_{W_i}=(\hat{S}_{W_i},\tilde{S}_{W_i})$ of keyword $W_i$ and the hidden structure's public part $\mathbf{Pub}=\hat{C}$, algorithm $\mathbf{StructuredSearch}(\mathbf{PK},\mathbf{Pub},\mathbb{C},T_{W_i})$ only finds out all ciphertexts of keyword $W_i$ with the hidden structure $\mathbf{Pub}$, where $\hat{S}_{W_i}=\mathbf{Extract}_{\mbox{\tiny IBKEM}}(\mathbf{SK}_{\mbox{\tiny IBKEM}},W_i)$, $\tilde{S}_{W_i}=\mathbf{Extract}_{\mbox{\tiny IBE}}(\mathbf{SK}_{\mbox{\tiny IBE}},$ $W_i)$, $\hat{C}$ is from $(\hat{K},\hat{C})=\mathbf{Encaps}_{\mbox{\tiny IBKEM}}(\mathbf{PK}_{\mbox{\tiny IBKEM}},W,u)$, keyword $W$ is arbitrarily chosen in $\mathcal{W}$, and $u$ is a random value. 

Algorithm $\mathbf{StructuredSearch}(\mathbf{PK},\mathbf{Pub},\mathbb{C},T_{W_i})$  computes $Pt^\prime=\mathbf{Decaps}_{\mbox{\tiny IBKEM}}(\hat{S}_{W_i},\mathbf{Pub})$ in its first step. According to the full-identity malleability of IBKEM in Definition \ref{D.FIM}, we have $\mathbf{FIM}(W_i,u)=\mathbf{Decaps}_{\mbox{\tiny IBKEM}}(\hat{S}_{W_i},\mathbf{Pub})$. So algorithm  $\mathbf{StructuredSearch}(\mathbf{PK},\mathbf{Pub},\mathbb{C},T_{W_i})$  finds out the ciphertext $(\mathbf{FIM}(W_i,u),\mathbf{Enc}_{\mbox{\tiny IBE}}(\mathbf{PK}_{\mbox{\tiny IBE}},W_i,Pt[u,W_i]))$ by matching $Pt^\prime$ with all ciphertexts' first part in its second step. Moreover, due to the collision-freeness of IBKEM in Definition \ref{D.FIM}, there is no keyword $W_j$ ($\neq W_i$) to meet $\mathbf{FIM}(W_i,u)=\mathbf{FIM}(W_j,u)$, and no hidden structure $\mathbf{Pub}^\prime$ ($\neq\mathbf{Pub}$) to meet $\mathbf{FIM}(W_i,u)=\mathbf{FIM}(W_i,u^\prime)$, where $\mathbf{Pub}^\prime$ is generated by algorithm $\mathbf{StructureInitialization}(\mathbf{PK})$ with the random value $u^\prime$. So only the ciphertext $(\mathbf{FIM}(W_i,u),\mathbf{Enc}_{\mbox{\tiny IBE}}(\mathbf{PK}_{\mbox{\tiny IBE}},W_i,Pt[u,W_i]))$ is found in this step, except with a negligible probability in the security parameter $k$. Then, according to the consistency of IBE, algorithm $\mathbf{StructuredSearch}(\mathbf{PK},\mathbf{Pub},\mathbb{C},T_{W_i})$ can decrypt $Pt[u,W_i]$ by algorithm $\mathbf{Dec}_{\mbox{\tiny IBE}}(\tilde{S}_{W_i},\mathbf{Enc}_{\mbox{\tiny IBE}}($ $\mathbf{PK}_{\mbox{\tiny IBE}},W_i,Pt[u,W_i]))$.

Recall that in algorithm $\mathbf{StructuredEncryption}$, $Pt[u,W_i]$ was randomly chosen in $\mathbb{G}_1$ and taken as the first part of only one ciphertext of keyword $W_i$. So when $\mathbf{StructuredSearch}(\mathbf{PK},$ $\mathbf{Pub},\mathbb{C},T_{W_i})$ goes back to its second step, only the ciphertext $(Pt[u,W_i],\mathbf{Enc}_{\mbox{\tiny IBE}}(\mathbf{PK}_{\mbox{\tiny IBE}},$ $W_i,R))$ is found, except with a negligible probability in the security parameter $k$. 

By carrying on in the same way, algorithm $\mathbf{StructuredSearch}($ $\mathbf{PK},\mathbf{Pub},\mathbb{C},T_{W_i})$ only finds out all ciphertexts of keyword $W_i$ with the hidden structure $\mathbf{Pub}$, except with a negligible probability in the security parameter $k$. And the algorithm will stop, since the random value $R$ contained in the last found ciphertext of keyword $W_i$ fails to  match any other ciphertext's first part.
\end{proof}

\subsection{Proof of Theorem \ref{T.SPCHS.Generic.Security}}\label{P.SPCHS.Generic.Security}

\begin{proof}
Let $\mathcal{G}_1$ and $\mathcal{G}_2$ be the challengers respectively in the Anon-SS-sID-CPA game of the underlying IBKEM scheme and the Anon-SS-ID-CPA game of the underlying IBE scheme. A constructed adversary $\mathcal{B}$ in the SS-sK-CKSA game of the generic SPCHS  construction is as follows.

\begin{itemize}
\item \textbf{Setup Phase}: In this phase,
\begin{enumerate}
\item $\mathcal{A}$ sends two challenge keywords $(W_0^*,W_1^*)$ to $\mathcal{B}$. 
\item $\mathcal{B}$ arbitrarily picks $I_1^*\leftarrow(\mathcal{ID}_{\mbox{\tiny IBKEM}}-\mathcal{W})$, and sends two challenge identities $(W_0^*,I_1^*)$ to $\mathcal{G}_1$. (The $I_1^*$ is existing, since we have $\mathcal{W}\subset\mathcal{ID}_{\mbox{\tiny IBKEM}}$.)
\item $\mathcal{G}_1$ generates  $(\mathbf{PK}_{\mbox{\tiny IBKEM}},\mathbf{SK}_{\mbox{\tiny IBKEM}})$ by algorithm $\mathbf{Setup}_{\mbox{\tiny IBKEM}}$ and sends $\mathbf{PK}_{\mbox{\tiny IBKEM}}$ to $\mathcal{B}$.
\item $\mathcal{B}$ queries $\mathcal{G}_1$ for the challenge key-and-encapsulation pair.
\item $\mathcal{G}_1$ picks $\hat{d}\overset{\$}\leftarrow\{0,1\}$, generates $(\hat{K}_0^*,\hat{C}_0^*)=\mathbf{Encaps}_{\mbox{\tiny IBKEM}}(\mathbf{PK}_{\mbox{\tiny IBKEM}},W_0^*,r_0)$ and $(\hat{K}_1^*,\hat{C}_1^*)=\mathbf{Encaps}_{\mbox{\tiny IBKEM}}(\mathbf{PK}_{\mbox{\tiny IBKEM}},I_1^*,r_1)$, and sends $(\hat{K}_{\hat{d}}^*,\hat{C}_0^*)$ to $\mathcal{B}$, where $r_0$ and $r_1$ are randomly chosen.
\item $\mathcal{B}$ adds $\hat{C}_0^*$ into the set $\mathbf{PSet}\subseteq\mathcal{C}_{\mbox{\tiny IBKEM}}$.
\item $\mathcal{G}_2$ generates $(\mathbf{PK}_{\mbox{\tiny IBE}},\mathbf{SK}_{\mbox{\tiny IBE}})$ by algorithm $\mathbf{Setup}_{\mbox{\tiny IBE}}$, and sends $\mathbf{PK}_{\mbox{\tiny IBE}}$ to $\mathcal{B}$.
\item $\mathcal{B}$ initializes the two lists $\mathbf{SList}=\emptyset\subseteq \mathcal{C}_{\mbox{\tiny IBKEM}}\times\{0,1\}^*$ and $\mathbf{Pt}=\emptyset\subseteq\mathcal{W}\times\mathcal{C}_{\mbox{\tiny IBKEM}}\times\mathcal{M}_{\mbox{\tiny IBE}}$, and initializes $N-1$ hidden structures by repeating the following steps for $i\in[1,N-1]$:
\begin{enumerate}
\item Pick a random value $u_i$ and an arbitrary keyword $W_i\in\mathcal{W}$;  
\item Generate $(\hat{K}_i,\hat{C}_i)=\mathbf{Encaps}_{\mbox{\tiny IBKEM}}(\mathbf{PK}_{\mbox{\tiny IBKEM}},W_i,u_i)$, add $\mathbf{Pub}_i=\hat{C}_i$ into the set $\mathbf{PSet}$, and add $(\mathbf{Pub}_i,u_i)$ into $\mathbf{SList}$;
\end{enumerate}
\item $\mathcal{B}$ finally sends $\mathbf{PK}$ and $\mathbf{PSet}$ to $\mathcal{A}$.
\end{enumerate}

\item \textbf{Query Phase 1}: In this phase, adversary $\mathcal{A}$ adaptively issues the following queries multiple times.
\begin{itemize}
\item Trapdoor Query $\mathcal{Q}_{Trap}(W)$: Taking as input a keyword $W\in\mathcal{W}$, $\mathcal{B}$ forwards the query $W$ both to the  decryption key oracles $\hat{S}_W=\mathcal{Q}_{DK}^{IBKEM}(W)$ and $\tilde{S}_{W}=\mathcal{Q}_{DK}^{IBE}(W)$, and sends $T_W=(\hat{S}_W,\tilde{S}_W)$ to $\mathcal{A}$. 

(In this query, $\mathcal{A}$ cannot query the keyword search trapdoor corresponding to the challenge keyword $W_0^*$ or $W_1^*$. In addition, one may find that $\mathcal{B}$ cannot respond the query $\mathcal{Q}_{Trap}(I^*_1)$. However, 
this is not a problem, since we let $I_1^*\in(\mathcal{ID}_{\mbox{\tiny IBKEM}}-\mathcal{W})$. So $\mathcal{A}$ never issues that query.)

\item Privacy Query $\mathcal{Q}_{Pri}(\mathbf{Pub})$: Taking as input a structure's public part $\mathbf{Pub}\in\mathbf{PSet}$, $\mathcal{B}$ aborts and outputs $\bot$ if $\mathbf{Pub}=\hat{C}_0^*$; otherwise, $\mathcal{B}$ retrieves $(\mathbf{Pub},u)$ from $\mathbf{SList}$ according to $\mathbf{Pub}$ and outputs $u$.
	
\item Encryption Query $\mathcal{Q}_{Enc}(W,\mathbf{Pub})$: Taking as inputs a keyword $W\in\mathcal{W}$ and a structure's public part $\mathbf{Pub}$, $\mathcal{B}$ does the following steps:
	\begin{enumerate}
	\item If $\mathbf{Pub}=\hat{C}_0^*\bigwedge W\neq W_0^*$, then 
	\begin{enumerate}
	\item Search $(W,\mathbf{Pub},Pt[u^*,W])$ for $W$ and $\mathbf{Pub}$ in $\mathbf{Pt}$;
	
	 (Note that $u^*$ is not a really known value. It is just a symbol to denote the random value used to generate $\mathbf{Pub}=\hat{C}_0^*$.)
	\item If it is not found, query $\hat{S}_W=\mathcal{Q}_{DK}^{IBKEM}(W)$, insert $(W,\mathbf{Pub},Pt[u^*,W]\overset{\$}\leftarrow\mathcal{M}_{\mbox{\tiny IBE}})$ to $\mathbf{Pt}$ and output $C=(\mathbf{Decaps}_{\mbox{\tiny IBKEM}}(\hat{S}_W,\mathbf{Pub}),\mathbf{Enc}_{\mbox{\tiny IBE}}(\mathbf{PK}_{\mbox{\tiny IBE}},W,Pt[u^*,W]))$;
	
	(Note that when $W=W_1^*$, $\mathcal{B}$ still can query $\hat{S}_W=\mathcal{Q}_{DK}^{IBKEM}(W)$, since $W_1^*$ is not a challenge IBKEM identity in the above \textbf{Setup Phase}. )
	\item Otherwise, pick $R\overset{\$}\leftarrow\mathcal{M}_{\mbox{\tiny IBE}}$, set $C=(Pt[u^*,W],\mathbf{Enc}_{\mbox{\tiny IBE}}(\mathbf{PK}_{\mbox{\tiny IBE}},W,R))$, update $Pt[u^*,W]=R$ and output $C$;
	\end{enumerate}
	\item If $\mathbf{Pub}=\hat{C}_0^*\bigwedge W= W_0^*$, then
	\begin{enumerate}
	\item Search $(W,\mathbf{Pub},Pt[u^*,W])$ for $W$ and $\mathbf{Pub}$ in $\mathbf{Pt}$;
	\item If it is not found, insert $(W,\mathbf{Pub},Pt[u^*,W]\overset{\$}\leftarrow\mathcal{M}_{\mbox{\tiny IBE}})$ to $\mathbf{Pt}$, and output $C=(\hat{K}_{\hat{d}}^*,\mathbf{Enc}_{\mbox{\tiny IBE}}(\mathbf{PK}_{\mbox{\tiny IBE}},W,Pt[u^*,W]))$;
	
	(Note that if $\hat{d}=0$, the output ciphertext $C$ is correct, since the full-identity malleability of the IBKEM scheme allows $\mathbf{FIM}(\hat{C}_0^*,W_0^*,u^*)=\hat{K}^*_{\hat{d}}$. Otherwise, the output ciphertext $C$ is incorrect. If $\mathcal{A}$ can find this incorrectness, it implies that $\hat{d}=1$ holds. Accordingly, $\mathcal{B}$ has advantage to win in the Anon-SS-sID-CPA game of the IBKEM scheme.)
	
	\item Otherwise, pick $R\overset{\$}\leftarrow\mathcal{M}_{\mbox{\tiny IBE}}$, set $C=(Pt[u^*,W],\mathbf{Enc}_{\mbox{\tiny IBE}}(\mathbf{PK}_{\mbox{\tiny IBE}},W,R))$, update $Pt[u^*,W]=R$ and output $C$;
	\end{enumerate}
	\item If $\mathbf{Pub}\neq\hat{C}_0^*$, then 
	\begin{enumerate}
	\item According to $\mathbf{Pub}$, retrieve $(\mathbf{Pub},u)$ from $\mathbf{SList}$;
	\item Search $(W,\mathbf{Pub},Pt[u,W])$ for $W$ and $\mathbf{Pub}$ in $\mathbf{Pt}$;
	\item If it is not found, insert $(W,\mathbf{Pub},Pt[u,W]\overset{\$}\leftarrow\mathcal{M}_{\mbox{\tiny IBE}})$ to $\mathbf{Pt}$ and output $C=(\mathbf{FIM}(W,u),\mathbf{Enc}_{\mbox{\tiny IBE}}(\mathbf{Pk}_{\mbox{\tiny IBE}},W,Pt[u,W]))$;
	\item Otherwise, pick $R\overset{\$}\leftarrow\mathcal{M}_{\mbox{\tiny IBE}}$, set $C=(Pt[u,W],\mathbf{Enc}_{\mbox{\tiny IBE}}(\mathbf{PK}_{\mbox{\tiny IBE}},W,R))$, update $Pt[u,W]=R$ and output $C$;
	\end{enumerate}
	\end{enumerate}
 \end{itemize}

\item \textbf{Challenge Phase}: In this phase, 
\begin{enumerate}
\item $\mathcal{A}$ sends two challenge structures $(\mathbf{Pub}^*_0,\mathbf{Pub}_1^*)\in\mathbf{PSet}\times\mathbf{PSet}$ to $\mathcal{B}$; 
\item $\mathcal{B}$ does the following steps:
\begin{enumerate}
\item If $\mathbf{Pub}_0^*\neq \hat{C}_0^*$, then abort and output $\bot$;
\item Send two challenge IBE identity-and-message pairs $(W_0^*,M_0^*)$ and $(I_1^*,M_1^*)$ to $\mathcal{G}_1$, where $M_0^*\overset{\$}\leftarrow\mathcal{M}_{\mbox{\tiny IBE}}$ and $M_1^*\overset{\$}\leftarrow\mathcal{M}_{\mbox{\tiny IBE}}$;
\end{enumerate}
\item $\mathcal{G}_2$ picks $\tilde{d}\overset{\$}\leftarrow\{0,1\}$, and sends the challenge IBE ciphertext $\tilde{C}^*_{\tilde{d}}=\mathbf{Enc}_{\mbox{\tiny IBE}}(\mathbf{PK}_{\mbox{\tiny IBE}},W^*_0,M^*_0)$ to $\mathcal{B}$ if $\tilde{d}=0$, otherwise sends $\tilde{C}_{\tilde{d}}=\mathbf{Enc}_{\mbox{\tiny IBE}}(\mathbf{PK}_{\mbox{\tiny IBE}},I^*_1,M_1^*)$ to $\mathcal{B}$.
\item $\mathcal{B}$ does the following steps:
\begin{enumerate}
\item Search $(W_0^*,\mathbf{Pub}_0^*,Pt[u^*,W_0^*])$ for $W_0^*$ and $\mathbf{Pub}_0^*$ in $\mathbf{Pt}$;
\item If it is not found, insert $(W_0^*,\mathbf{Pub}_0^*,Pt[u^*,W_0^*]=M_0^*)$ to $\mathbf{Pt}$, output the challenge ciphertext $C^*$ to $\mathcal{A}$ and stop this phase, where $C^*=(\hat{K}_{\hat{d}}^*,\tilde{C}_{\tilde{d}})$;

(Note that if $\hat{d}=0$ and $\tilde{d}=0$, the $C^*$ is a correct one. Otherwise, it is an incorrect one. If $\mathcal{A}$ confirms the incorrectness of $C^*$, it implies that $\hat{d}=1$ or $\tilde{d}=1$ holds. Accordingly, $\mathcal{B}$ has advantage to win in the Anon-SS-sID-CPA game of the IBKEM or the Anon-SS-ID-CPA game of the IBE scheme.)

\item Otherwise, set the challenge ciphertext $C^*=(Pt[u^*,W_0^*],\tilde{C}_{\tilde{d}})$, update $Pt[u^*,W_0^*]=M_0^*$, send $C^*$ to $\mathcal{A}$ and stop this phase. 

(Note that  if $\tilde{d}=0$, the $C^*$ is a correct one. Otherwise, it is an incorrect one. If $\mathcal{A}$ confirms  the incorrectness of $C^*$, it implies that $\tilde{d}=1$ holds. Accordingly, $\mathcal{B}$ has advantage to win in the Anon-SS-ID-CPA game of the IBE scheme.)
\end{enumerate}
\end{enumerate}

\item \textbf{Query Phase 2}: This phase is the same as \textbf{Query Phase 1}. Note that in \textbf{Query Phase 1} and \textbf{Query Phase 2}, adversary $\mathcal{A}$ cannot query the private part corresponding to  the structure  $\mathbf{Pub}^*_0$ or $\mathbf{Pub}^*_1$ and the keyword search trapdoor corresponding to the challenge keyword $W^*_0$ or $W^*_1$.

\item \textbf{Guess Phase}: Adversary $\mathcal{A}$ sends a guess $d^\prime$ to adversary $\mathcal{B}$. $\mathcal{B}$ takes $d^\prime$ as his guess at both $\hat{d}$ and $\tilde{d}$, and forwards $d^\prime$ to challengers $\mathcal{G}_1$ and $\mathcal{G}_2$.
\end{itemize}

Let $\overline{Abort}$ denote the event that adversary $\mathcal{B}$ does not abort in the above game. Suppose adversary $\mathcal{A}$ totally queries $\mathcal{Q}_{Pri}$ for $q_p$ times. Then we have $Pr[\overline{Abort}]=\frac{N-q_p}{N}\cdot \frac{1}{N-q_p}=\frac{1}{N}$. Note that $q_p\leq(N-2)$ always holds, since adversary $\mathcal{A}$ cannot query $\mathcal{Q}_{Pri}$ for the challenge structures $(\mathbf{Pub}^*_0,\mathbf{Pub}_1^*)$. 

Let $Win_{IBKEM,\mathcal{B}}^{\text{Anon-SS-sID-CPA}}$ denote the event that  $\mathcal{B}$ wins in the Anon-SS-sID-CPA game of the underlying IBKEM scheme under the condition that $\mathcal{B}$ does not abort. Let $Win_{IBE,\mathcal{B}}^{\text{Anon-SS-ID-CPA}}$ denote the event that  $\mathcal{B}$ wins in the Anon-SS-ID-CPA game of the underlying IBE scheme under the condition that $\mathcal{B}$ does not abort. Let $Adv_{\mathcal{B}}$ be the advantage of $\mathcal{B}$ to have $Win_{IBKEM,\mathcal{B}}^{\text{Anon-SS-sID-CPA}}$ or $Win_{IBE,\mathcal{B}}^{\text{Anon-SS-ID-CPA}}$ holds. Since $\mathcal{B}$ has the probability no less than $\frac{3}{4}$ to have $Win_{IBKEM,\mathcal{B}}^{\text{Anon-SS-sID-CPA}}$ or $Win_{IBE,\mathcal{B}}^{\text{Anon-SS-ID-CPA}}$ holds under the condition that $\mathcal{B}$ does not abort, we clearly have 

\begin{equation}
\begin{aligned}
Adv_{\mathcal{B}}&=(Pr[Win_{IBKEM,\mathcal{B}}^{\text{Anon-SS-sID-CPA}}\bigvee Win_{IBE,\mathcal{B}}^{\text{Anon-SS-ID-CPA}}|\overline{Abort}]-\frac{3}{4})\cdot Pr[\overline{Abort}]\\
&=(Pr[Win_{IBKEM,\mathcal{B}}^{\text{Anon-SS-sID-CPA}}|\overline{Abort}]+Pr[Win_{IBE,\mathcal{B}}^{\text{Anon-SS-ID-CPA}}|\overline{Abort}]\\ &\qquad-Pr[Win_{IBKEM,\mathcal{B}}^{\text{Anon-SS-sID-CPA}}\bigwedge Win_{IBE,\mathcal{B}}^{\text{Anon-SS-ID-CPA}}|\overline{Abort}]-\frac{3}{4})\cdot Pr[\overline{Abort}]\nonumber
\end{aligned}
\end{equation}

Let $\overline{Belong}$ denote the event that $(W_0^*,\mathbf{Pub}_0^*,Pt[u^*,W_0^*])\notin\mathbf{Pt}$ holds in the above \textbf{Challenge Phase}. On the contrary, let $Belong$ denote the event that $(W_0^*,\mathbf{Pub}_0^*,Pt[u^*,W_0^*])\in\mathbf{Pt}$ holds in the above \textbf{Challenge Phase}.

We compute the probability $Pr[Win_{IBKEM,\mathcal{B}}^{\text{Anon-SS-sID-CPA}}|\overline{Abort}]+Pr[Win_{IBE,\mathcal{B}}^{\text{Anon-SS-ID-CPA}}|\overline{Abort}]$ as follows.
\begin{equation}
\begin{aligned}
&Pr[Win_{IBKEM,\mathcal{B}}^{\text{Anon-SS-sID-CPA}}|\overline{Abort}]+Pr[Win_{IBE,\mathcal{B}}^{\text{Anon-SS-ID-CPA}}|\overline{Abort}]\\
&=Pr[d^\prime=\hat{d}|\overline{Abort}\bigwedge\overline{Belong}]\cdot Pr[\overline{Belong}]+Pr[d^\prime=\hat{d}|\overline{Abort}\bigwedge Belong]\cdot Pr[Belong]\\
&\qquad +Pr[d^\prime=\tilde{d}|\overline{Abort}\bigwedge\overline{Belong}]\cdot Pr[\overline{Belong}]+Pr[d^\prime=\tilde{d}|\overline{Abort}\bigwedge Belong]\cdot Pr[Belong]\\
&=(Pr[d^\prime=\hat{d}|\overline{Abort}\bigwedge\overline{Belong}\bigwedge\hat{d}=0\bigwedge\tilde{d}=0]\cdot Pr[\hat{d}=0\bigwedge\tilde{d}=0]\\
&\qquad +Pr[d^\prime=\hat{d}|\overline{Abort}\bigwedge\overline{Belong}\bigwedge\hat{d}=1\bigwedge\tilde{d}=0]\cdot Pr[\hat{d}=1\bigwedge\tilde{d}=0]\\
&\qquad + Pr[d^\prime=\hat{d}|\overline{Abort}\bigwedge\overline{Belong}\bigwedge\hat{d}=0\bigwedge\tilde{d}=1]\cdot Pr[\hat{d}=0\bigwedge\tilde{d}=1]\\
&\qquad +Pr[d^\prime=\hat{d}|\overline{Abort}\bigwedge\overline{Belong}\bigwedge\hat{d}=1\bigwedge\tilde{d}=1]\cdot Pr[\hat{d}=1\bigwedge\tilde{d}=1])\cdot Pr[\overline{Belong}]\\
&\qquad +(Pr[d^\prime=\hat{d}|\overline{Abort}\bigwedge Belong\bigwedge\hat{d}=0\bigwedge\tilde{d}=0]\cdot Pr[\hat{d}=0\bigwedge\tilde{d}=0]\\
&\qquad +Pr[d^\prime=\hat{d}|\overline{Abort}\bigwedge Belong\bigwedge\hat{d}=1\bigwedge\tilde{d}=0]\cdot Pr[\hat{d}=1\bigwedge\tilde{d}=0]\\
&\qquad + Pr[d^\prime=\hat{d}|\overline{Abort}\bigwedge Belong\bigwedge\hat{d}=0\bigwedge\tilde{d}=1]\cdot Pr[\hat{d}=0\bigwedge\tilde{d}=1]\\
&\qquad +Pr[d^\prime=\hat{d}|\overline{Abort}\bigwedge Belong\bigwedge\hat{d}=1\bigwedge\tilde{d}=1]\cdot Pr[\hat{d}=1\bigwedge\tilde{d}=1])\cdot Pr[Belong]\\
&\qquad +(Pr[d^\prime=\tilde{d}|\overline{Abort}\bigwedge\overline{Belong}\bigwedge\hat{d}=0\bigwedge\tilde{d}=0]\cdot Pr[\hat{d}=0\bigwedge\tilde{d}=0]\\
&\qquad +Pr[d^\prime=\tilde{d}|\overline{Abort}\bigwedge\overline{Belong}\bigwedge\hat{d}=1\bigwedge\tilde{d}=0]\cdot Pr[\hat{d}=1\bigwedge\tilde{d}=0]\\
&\qquad + Pr[d^\prime=\tilde{d}|\overline{Abort}\bigwedge\overline{Belong}\bigwedge\hat{d}=0\bigwedge\tilde{d}=1]\cdot Pr[\hat{d}=0\bigwedge\tilde{d}=1]\\
&\qquad +Pr[d^\prime=\tilde{d}|\overline{Abort}\bigwedge\overline{Belong}\bigwedge\hat{d}=1\bigwedge\tilde{d}=1]\cdot Pr[\hat{d}=1\bigwedge\tilde{d}=1])\cdot Pr[\overline{Belong}]\\
&\qquad +(Pr[d^\prime=\tilde{d}|\overline{Abort}\bigwedge Belong\bigwedge\hat{d}=0\bigwedge\tilde{d}=0]\cdot Pr[\hat{d}=0\bigwedge\tilde{d}=0]\\
&\qquad +Pr[d^\prime=\tilde{d}|\overline{Abort}\bigwedge Belong\bigwedge\hat{d}=1\bigwedge\tilde{d}=0]\cdot Pr[\hat{d}=1\bigwedge\tilde{d}=0]\\
&\qquad + Pr[d^\prime=\tilde{d}|\overline{Abort}\bigwedge Belong\bigwedge\hat{d}=0\bigwedge\tilde{d}=1]\cdot Pr[\hat{d}=0\bigwedge\tilde{d}=1]\\
&\qquad +Pr[d^\prime=\tilde{d}|\overline{Abort}\bigwedge Belong\bigwedge\hat{d}=1\bigwedge\tilde{d}=1]\cdot Pr[\hat{d}=1\bigwedge\tilde{d}=1])\cdot Pr[Belong]\\
&=(2\cdot Adv^{\text{SS-sK-CKSA}}_{SPCHS,\mathcal{A}}+3+Pr[d^\prime=\hat{d}|\overline{Abort}\bigwedge\overline{Belong}\bigwedge\hat{d}=1\bigwedge\tilde{d}=1]\\&\qquad +Pr[d^\prime=\tilde{d}|\overline{Abort}\bigwedge\overline{Belong}\bigwedge\hat{d}=1\bigwedge\tilde{d}=1])\cdot \frac{1}{4}\cdot Pr[\overline{Belong}]\\&\qquad +(2\cdot Adv^{\text{SS-sK-CKSA}}_{SPCHS,\mathcal{A}}+3+Pr[d^\prime=\hat{d}|\overline{Abort}\bigwedge Belong\bigwedge\hat{d}=1\bigwedge\tilde{d}=1]\\&\qquad+Pr[d^\prime=\tilde{d}|\overline{Abort}\bigwedge Belong\bigwedge\hat{d}=1\bigwedge\tilde{d}=1])\cdot \frac{1}{4}\cdot Pr[Belong]\\&= (2\cdot Adv^{\text{SS-sK-CKSA}}_{SPCHS,\mathcal{A}}+3+2\cdot Pr[d^\prime=\hat{d}=\tilde{d}|\overline{Abort}\bigwedge\overline{Belong}\bigwedge\hat{d}=1\bigwedge\tilde{d}=1])\cdot \frac{1}{4}\cdot Pr[\overline{Belong}]\\&\qquad +(2\cdot Adv^{\text{SS-sK-CKSA}}_{SPCHS,\mathcal{A}}+3+2\cdot Pr[d^\prime=\hat{d}=\tilde{d}|\overline{Abort}\bigwedge Belong\bigwedge\hat{d}=1\bigwedge\tilde{d}=1])\cdot \frac{1}{4}\cdot Pr[Belong]\\&=(2\cdot Adv^{\text{SS-sK-CKSA}}_{SPCHS,\mathcal{A}}+3)\cdot\frac{1}{4}+2\cdot Pr[d^\prime=\hat{d}=\tilde{d}|\overline{Abort}\bigwedge\hat{d}=1\bigwedge\tilde{d}=1]\cdot \frac{1}{4}\\&=\frac{1}{2}\cdot Adv^{\text{SS-sK-CKSA}}_{SPCHS,\mathcal{A}}+1\nonumber
\end{aligned}
\end{equation}

We compute the probability $Pr[Win_{IBKEM,\mathcal{B}}^{\text{Anon-SS-sID-CPA}}\bigwedge Win_{IBE,\mathcal{B}}^{\text{Anon-SS-ID-CPA}}|\overline{Abort}]$ as follows.

\begin{equation}
\begin{aligned}
&Pr[Win_{IBKEM,\mathcal{B}}^{\text{Anon-SS-sID-CPA}}\bigwedge Win_{IBE,\mathcal{B}}^{\text{Anon-SS-ID-CPA}}|\overline{Abort}]\\
&=Pr[d^\prime=\hat{d}=\tilde{d}|\overline{Abort}\bigwedge\overline{Belong}]\cdot Pr[\overline{Belong}]+Pr[d^\prime=\hat{d}=\tilde{d}|\overline{Abort}\bigwedge Belong]\cdot Pr[Belong]\\
&=(Pr[d^\prime=\hat{d}=\tilde{d}|\overline{Abort}\bigwedge\overline{Belong}\bigwedge\hat{d}=0\bigwedge\tilde{d}=0]\cdot Pr[\hat{d}=0\bigwedge\tilde{d}=0]\\
&\qquad +Pr[d^\prime=\hat{d}=\tilde{d}|\overline{Abort}\bigwedge\overline{Belong}\bigwedge\hat{d}=1\bigwedge\tilde{d}=1]\cdot Pr[\hat{d}=1\bigwedge\tilde{d}=1])\cdot Pr[\overline{Belong}]\\
&\qquad +(Pr[d^\prime=\hat{d}=\tilde{d}|\overline{Abort}\bigwedge Belong\bigwedge\hat{d}=0\bigwedge\tilde{d}=0]\cdot Pr[\hat{d}=0\bigwedge\tilde{d}=0]\\
&\qquad +Pr[d^\prime=\hat{d}=\tilde{d}|\overline{Abort}\bigwedge Belong\bigwedge\hat{d}=1\bigwedge\tilde{d}=1]\cdot Pr[\hat{d}=1\bigwedge\tilde{d}=1])\cdot Pr[Belong]\\
&=(Adv^{\text{SS-sK-CKSA}}_{SPCHS,\mathcal{A}}+\frac{1}{2}+Pr[d^\prime=\hat{d}=\tilde{d}|\overline{Abort}\bigwedge\overline{Belong}\bigwedge\hat{d}=1\bigwedge\tilde{d}=1])\cdot\frac{1}{4}\cdot Pr[\overline{Belong}]\\&\qquad +(Adv^{\text{SS-sK-CKSA}}_{SPCHS,\mathcal{A}}+\frac{1}{2}+Pr[d^\prime=\hat{d}=\tilde{d}|\overline{Abort}\bigwedge Belong\bigwedge\hat{d}=1\bigwedge\tilde{d}=1])\cdot\frac{1}{4}\cdot Pr[Belong]\\
&=(Adv^{\text{SS-sK-CKSA}}_{SPCHS,\mathcal{A}}+\frac{1}{2})\cdot \frac{1}{4}+Pr[d^\prime=\hat{d}=\tilde{d}|\overline{Abort}\bigwedge\hat{d}=1\bigwedge\tilde{d}=1]\cdot\frac{1}{4}=\frac{1}{4}\cdot Adv^{\text{SS-sK-CKSA}}_{SPCHS,\mathcal{A}}+\frac{1}{4}\nonumber
\end{aligned}
\end{equation}

According to the above computations, we have 

\begin{equation}
\begin{aligned}
Adv_{\mathcal{B}}&=(Pr[Win_{IBKEM,\mathcal{B}}^{\text{Anon-SS-sID-CPA}}\bigvee Win_{IBE,\mathcal{B}}^{\text{Anon-SS-ID-CPA}}|\overline{Abort}]-\frac{3}{4})\cdot Pr[\overline{Abort}]\\
&=(Pr[Win_{IBKEM,\mathcal{B}}^{\text{Anon-SS-sID-CPA}}|\overline{Abort}]+Pr[Win_{IBE,\mathcal{B}}^{\text{Anon-SS-ID-CPA}}|\overline{Abort}]\\ &\qquad-Pr[Win_{IBKEM,\mathcal{B}}^{\text{Anon-SS-sID-CPA}}\bigwedge Win_{IBE,\mathcal{B}}^{\text{Anon-SS-ID-CPA}}|\overline{Abort}]-\frac{3}{4})\cdot Pr[\overline{Abort}]\\&=\frac{1}{4N}\cdot Adv^{\text{SS-sK-CKSA}}_{SPCHS,\mathcal{A}}\nonumber
\end{aligned}
\end{equation}

In addition, it is clear that adversary $\mathcal{B}$ is a PPT adversary, if $\mathcal{A}$ is a PPT adversary. In conclusion, we have that if a PPT adversary $\mathcal{A}$ wins in the SS-sK-CKSA game of the generic SPCHS construction with advantage $Adv^{\text{SS-sK-CKSA}}_{SPCHS,\mathcal{A}}$, then the above PPT adversary $\mathcal{B}$ can utilize the capability of adversary $\mathcal{A}$ to win in the Anon-SS-sID-CPA game of the underlying IBKEM scheme or the Anon-SS-ID-CPA game of the underlying IBE scheme with advantage $\frac{1}{4N}\cdot Adv^{\text{SS-sK-CKSA}}_{SPCHS,\mathcal{A}}$.
\end{proof}

\subsection{A Collision-free Full-identity Malleable IBKEM Instance in the RO Model}\label{AS.IBKEM.Instance.1}

We first review the VRF-suitable IBKEM instance proposed in Appendix A.2 of \cite{ACF13}. Then we prove its collision-free full-identity malleability and the Anon-SS-ID-CPA security in the RO model. Let identity space $\mathcal{ID}_{\mbox{\tiny IBKEM}}=\{0,1\}^*$. This IBKEM instance is as follows.

\begin{itemize} 
\item $\mathbf{Setup}_{\mbox{\tiny IBKEM}}(1^k,\mathcal{ID}_{\mbox{\tiny IBKEM}})$: Take as input a security parameter $1^k$ and the identity space $\mathcal{ID}_{\mbox{\tiny IBKEM}}$, compute $(q,\mathbb{G}, \mathbb{G}_1, g,\hat{e})\overset{\$}\leftarrow \mathbf{BGen}(1^k)$, pick $s\overset{\$}\leftarrow \mathbb{Z}_q^*$, set $P\leftarrow g^s$, choose a cryptographic hash function $H: \{0,1\}^*\rightarrow \mathbb{G}$, set the encapsulated key space $\mathcal{K}_{\mbox{\tiny IBKEM}}=\mathbb{G}_1$, set the encapsulation space $\mathcal{C}_{\mbox{\tiny IBKEM}}=\mathbb{G}$, and output the master public key $\mathbf{PK}_{\mbox{\tiny IBKEM}}=(q,\mathbb{G},\mathbb{G}_1,g,\hat{e},P,H,\mathcal{ID}_{\mbox{\tiny IBKEM}},\mathcal{K}_{\mbox{\tiny IBKEM}},\mathcal{C}_{\mbox{\tiny IBKEM}})$ and the master secret key $\mathbf{SK}_{\mbox{\tiny IBKEM}}=s$.

\item $\mathbf{Extract}_{\mbox{\tiny IBKEM}}(\mathbf{SK}_{\mbox{\tiny IBKEM}},ID)$: Take as inputs $\mathbf{SK}_{\mbox{\tiny IBKEM}}$ and an identity $ID\in\mathcal{ID}_{\mbox{\tiny IBKEM}}$, and output a decryption key $\hat{S}_{ID}=H(ID)^s$ of $ID$.

\item $\mathbf{Encaps}_{\mbox{\tiny IBKEM}}(\mathbf{PK}_{\mbox{\tiny IBKEM}},ID,r)$: Take as inputs $\mathbf{PK}_{\mbox{\tiny IBKEM}}$, an identity $ID\in\mathcal{ID}_{\mbox{\tiny IBKEM}}$ and a random value $r$, and output a key-and-encapsulation pair $(\hat{K},\hat{C})$, where $\hat{K}=\hat{e}(P,H(ID))^r$ and $\hat{C}=g^r$.

\item $\mathbf{Decaps}_{\mbox{\tiny IBKEM}}(\hat{S}_{ID^\prime},\hat{C})$: Take as inputs the decryption key $\hat{S}_{ID^\prime}$ of identity $ID^\prime$ and an encapsulation $\hat{C}$, and output the encapsulated key $\hat{K}=\hat{e}(\hat{C},\hat{S}_{ID^\prime})$ if $\hat{C}\in\mathbb{G}$ or output $\bot$ otherwise.
\end{itemize}

\textbf{Collision-Free Full-Identity Malleability.} Let the function $\mathbf{FIM}(ID,r)=\hat{e}(P,H(ID))^r$ for any identity $ID\in\mathcal{ID}_{\mbox{\tiny IBKEM}}$ and any random value $r\in\mathbb{Z}_q^*$. Clearly, the function $\mathbf{FIM}$ is efficient. Moreover, it is easy to find that the function $\mathbf{FIM}$ has collision-freeness and full-identity malleability by the following reasons. 

For any $(\hat{K},\hat{C})=\mathbf{Encaps}_{\mbox{\tiny IBKEM}}(\mathbf{PK}_{\mbox{\tiny IBKEM}},ID,r)$ and any identity $ID^\prime\in\mathcal{ID}_{\mbox{\tiny IBKEM}}$, it is clear that $\mathbf{FIM}(ID^\prime,r)=\hat{e}(P,H(ID^\prime))^r=\mathbf{Decaps}_{\mbox{\tiny IBKEM}}(\hat{S}_{ID^\prime},\hat{C})$ holds. So the function $\mathbf{FIM}$ has full-identity malleability. In addition, for any identity $ID^\prime\in\mathcal{ID}_{\mbox{\tiny IBKEM}}$, if $ID\neq ID^\prime$, we clearly have $\mathbf{FIM}(ID,r)\neq \mathbf{FIM}(ID^\prime,r)$ due to the collision freeness of the hash function $H$; for any random value $r^\prime\in\mathbb{Z}_q^*$, if $r\neq r^\prime$, we clearly have $\mathbf{FIM}(ID,r)\neq\mathbf{FIM}(ID,r^\prime)$ due to the randomness of the values $r$ and $r^\prime$. Therefore, the function $\mathbf{FIM}$ offers collision-freeness,  except with a negligible probability in the security parameter $k$. 

\textbf{Anon-SS-ID-CPA Security.} The Anon-SS-ID-CPA security of the above IBKEM instance is based on the DBDH assumption in the RO model. The formal result is the following theorem. 

\begin{theorem}
Let the hash function $H$ be modeled as the random oracle $\mathcal{Q}_{H}(\cdot)$. Suppose a PPT adversary $\mathcal{A}$ wins in the Anon-SS-ID-CPA game of the above IBKEM instance with advantage $Adv^{\text{Anon-SS-ID-CPA}}_{IBKEM,\mathcal{A}}$, in which
$\mathcal{A}$ makes at most $q_p$ queries to oracle $\mathcal{Q}_{DK}^{IBKEM}(\cdot)$. Then there is a PPT algorithm $\mathcal{B}$ that solves the
DBDH problem in $\mathbf{BGen}(1^k)$ with advantage approximately
$$Adv_\mathcal {B}^{DBDH}(1^k)\approx\frac{4}{(e\cdot q_p)^2}\cdot Adv^{\text{Anon-SS-ID-CPA}}_{IBKEM,\mathcal{A}}$$ where $e$ is the base of natural logarithms.
\end{theorem}

\begin{proof}
To prove this theorem, we will construct a PPT algorithm $\mathcal{B}$ that plays the Anon-SS-ID-CPA game with adversary $\mathcal{A}$ and utilizes the capability of $\mathcal{A}$ to solve the DBDH problem in $\mathbf{BGen}(1^k)$ with advantage approximately $\frac{4}{(e\cdot q_p)^2}\cdot Adv^{\text{Anon-SS-ID-CPA}}_{IBKEM,\mathcal{A}} $. Let $Coin\overset{\sigma}\leftarrow \{0,1\}$ denote the operation that picks $Coin\in\{0,1\}$ according to the probability $Pr[Coin = 1] = \sigma$ (the specified value of $\sigma$ will be decided latter). The constructed algorithm $\mathcal{B}$ in the Anon-SS-ID-CPA game is as follows.

\begin{itemize}
\item \textbf{Setup Phase}: Algorithm $\mathcal{B}$ takes as inputs $(q,\mathbb{G},\mathbb{G}_1,g,\hat{e},g^a,g^b,g^c,Z)$ (where $Z$ equals either $\hat{e}(g,g)^{abc}$ or $\hat{e}(g,g)^y$) and the identity space $\mathcal{ID}_{\mbox{\tiny IBKEM}}$, and does the following steps:
\begin{enumerate}
\item Initialize a list $\mathbf{HList}=\emptyset\subseteq\mathcal{ID}_{\mbox{\tiny IBKEM}}\times\mathbb{G}\times\mathbb{Z}_q^*\times \{0,1\}$;
\item Set the encapsulated key space $\mathcal{K}_{\mbox{\tiny IBKEM}}=\mathbb{G}_1$, the encapsulation space $\mathcal{C}_{\mbox{\tiny IBKEM}}=\mathbb{G}$ and $\mathbf{PK}_{\mbox{\tiny IBKEM}}=(q,\mathbb{G},\mathbb{G}_1,g,\hat{e},P=g^a,\mathcal{ID}_{\mbox{\tiny IBKEM}},\mathcal{K}_{\mbox{\tiny IBKEM}},\mathcal{C}_{\mbox{\tiny IBKEM}})$;
\item Send $\mathbf{PK}_{\mbox{\tiny IBKEM}}$ to adversary $\mathcal{A}$;
\end{enumerate}

\item \textbf{Query Phase 1}: Adversary $\mathcal{A}$ adaptively issues the following queries multiple times.
\begin{itemize}
\item Hash Query $\mathcal{Q}_{H}(ID)$: Taking as input an identity $ID\in\mathcal{ID}_{\mbox{\tiny IBKEM}}$, algorithm $\mathcal{B}$ does the following steps:
	\begin{enumerate}
	\item Pick $x\overset{\$}\leftarrow\mathbb{Z}_q^*$ and $Coin\overset{\sigma}\leftarrow\{0,1\}$;
	\item If $Coin=0$, add $(ID,g^x,x,Coin)$ into $\mathbf{HList}$ and output $g^x$;
	\item Otherwise, add $(ID,g^{c\cdot x},x,Coin)$ into $\mathbf{HList}$ and output $g^{c\cdot x}$;
	\end{enumerate}
\item Decryption Key Query $\mathcal{Q}_{DK}^{IBKEM}(ID)$: Taking as input an identity $ID\in\mathcal{ID}_{\mbox{\tiny IBKEM}}$, algorithm $\mathcal{B}$ does the following steps:
	\begin{enumerate}
	\item If $(ID,*,*,*)\notin\mathbf{HList}$, query $\mathcal{Q}_H(ID)$;
	\item According to $ID$, retrieve $(ID,X,x,Coin)$ from $\mathbf{HList}$;
	\item If $Coin=0$, output $g^{a\cdot x}$; otherwise, abort and output $\bot$;
	\end{enumerate}
 \end{itemize}

\item \textbf{Challenge Phase}: Adversary $\mathcal{A}$ sends two challenge identities $ID^*_0\in\mathcal{ID}_{\mbox{\tiny IBKEM}}$ and $ID_1^*\in\mathcal{ID}_{\mbox{\tiny IBKEM}}$ to algorithm $\mathcal{B}$;  $\mathcal{B}$ picks $\hat{d}\overset{\$}\leftarrow\{0,1\}$, and does the following steps:
\begin{enumerate}
\item If $(ID_0^*,*,*,*)\notin\mathbf{HList}$, query $\mathcal{Q}_H(ID_0^*)$;
\item If $(ID_1^*,*,*,*)\notin\mathbf{HList}$, query $\mathcal{Q}_H(ID_1^*)$;
\item According to $ID_0^*$ and $ID_1^*$, retrieve $(ID_0^*,X_0^*,x_0^*,Coin_0^*)$ and $(ID_1^*,X_1^*,x_1^*,Coin_1^*)$ from $\mathbf{HList}$; 
\item If $Coin_0^*=0\bigvee Coin_1^*=0$, then abort and output $\bot$;
\item Finally send the challenge key-and-encapsulation pair $(Z^{x_{\hat{d}}^*},g^b)$ to adversary $\mathcal{A}$; 
\end{enumerate}

\item \textbf{Query Phase 2}: This phase is the same as \textbf{Query Phase 2}. Note that in \textbf{Query Phase 1} and \textbf{Query Phase 2}, adversary $\mathcal{A}$ cannot query the decryption key corresponding to the challenge identity $ID_0^*$ or $ID_1^*$.

\item \textbf{Guess Phase}: Adversary $\mathcal{A}$ sends a guess $\hat{d}^\prime$ to algorithm $\mathcal{B}$. If $\hat{d}=\hat{d}^\prime$, $\mathcal{B}$ output 1; otherwise, output 0.
\end{itemize}

Let $\overline{Abort}$ denote the event that algorithm $\mathcal{B}$ does not abort in the above game. Next, we will compute the probabilities $Pr[\overline{Abort}]$, $Pr[\mathcal{B}=1|Z=\hat{e}(g,g)^{abc}]$ and $Pr[\mathcal{B}=1|Z=\hat{e}(g,g)^y]$, and the advantage $Adv_\mathcal {B}^{DBDH}(1^k)$.

According to the above game, the probability of the event $\overline{Abort}$ only relies on the probability $\sigma$ and the number of times of adversary $\mathcal{A}$ to query oracle $\mathcal{Q}_{DK}^{IBKEM}(ID)$. We have that $Pr[\overline{Abort}]=(1-\sigma)^{q_p}\cdot \sigma^2$. Let $\sigma=\frac{2}{2+q_p}$. We have that  $Pr[\overline{Abort}]\approx\frac{4}{(e\cdot q_p)^2}$, where $e$ is the base of natural logarithms.

When $Z=\hat{e}(g,g)^{abc}$ and the event $\overline{Abort}$ holds, it is easy to find that algorithm $\mathcal{B}$ simulates a real Anon-SS-ID-CPA game in adversary $\mathcal{A}$'s mind. So we have $Pr[\hat{d}=\hat{d}^\prime|\overline{Abort}\bigwedge Z=\hat{e}(g,g)^{abc}]=(Adv^{\text{Anon-SS-ID-CPA}}_{IBKEM,\mathcal{A}}+\frac{1}{2})$.
 
When $Z=\hat{e}(g,g)^y$ and the event $\overline{Abort}$ holds, algorithm $\mathcal{B}$ generates an incorrect challenge ciphertext, and it is independent of the challenge identities $ID_0^*$ and $ID_1^*$. So we have $Pr[\hat{d}=\hat{d}^\prime|\overline{Abort}\bigwedge Z=\hat{e}(g,g)^y]=\frac{1}{2}$.

Now, we can compute the advantage $Adv_\mathcal{B}^{DBDH}(1^k)$ as follows:
\begin{equation}
\begin{aligned}
Adv_\mathcal{B}^{DBDH}(1^k)&=Pr[\mathcal{B}=1|Z=\hat{e}(g,g)^{abc}]-Pr[\mathcal{B}=1|Z=\hat{e}(g,g)^y]\\
&=Pr[\hat{d}=\hat{d}^\prime\bigwedge\overline{Abort}|Z=\hat{e}(g,g)^{abc}]-Pr[\hat{d}=\hat{d}^\prime\bigwedge\overline{Abort}|Z=\hat{e}(g,g)^y]\\
&=Pr[\hat{d}=\hat{d}^\prime|\overline{Abort}\bigwedge Z=\hat{e}(g,g)^{abc}]\cdot Pr[\overline{Abort}|Z=\hat{e}(g,g)^{abc}]\\&\qquad -Pr[\hat{d}=\hat{d}^\prime|\overline{Abort}\bigwedge Z=\hat{e}(g,g)^y]\cdot Pr[\overline{Abort}|Z=\hat{e}(g,g)^y]\\
&\approx(Adv^{\text{Anon-SS-ID-CPA}}_{IBKEM,\mathcal{A}}+\frac{1}{2})\cdot\frac{4}{(e\cdot  q_p)^2}-\frac{1}{2}\cdot\frac{4}{(e\cdot q_p)^2}=\frac{4}{(e\cdot q_p)^2}\cdot Adv^{\text{Anon-SS-ID-CPA}}_{IBKEM,\mathcal{A}} \nonumber
\end{aligned}
\end{equation}

In addition, it is clear that algorithm $\mathcal{B}$ is a PPT algorithm, if adversary $\mathcal{A}$ is a PPT adversary. In conclusion, if a PPT adversary $\mathcal{A}$ wins in the Anon-SS-ID-CPA game of the above IBKEM instance with advantage $Adv^{\text{Anon-SS-ID-CPA}}_{IBKEM,\mathcal{A}}$, in which
$\mathcal{A}$ makes at most $q_p$ queries to oracle $\mathcal{Q}_{DK}^{IBKEM}(\cdot)$, then there is a PPT algorithm $\mathcal{B}$ that solves the
DBDH problem in $\mathbf{BGen}(1^k)$ with advantage approximately
$$Adv_\mathcal {B}^{DBDH}(1^k)\approx\frac{4}{(e\cdot q_p)^2}\cdot Adv^{\text{Anon-SS-ID-CPA}}_{IBKEM,\mathcal{A}}$$ where $e$ is the base of natural logarithms.
\end{proof}

\subsection{Proof of Theorem \ref{T.IBKEM.Instance.Security}}\label{P.IBKEM.Instance.Security}

\begin{proof}
Suppose a PPT adversary $\mathcal{A}$ wins in the Anon-SS-ID-CPA game of the above IBKEM instance with advantage $Adv^{\text{Anon-SS-ID-CPA}}_{IBKEM,\mathcal{A}}$, in which
$\mathcal{A}$ makes at most $q_p$ queries to oracle $\mathcal{Q}_{DK}^{IBKEM}(\cdot)$. To prove this theorem, we will construct a PPT algorithm $\mathcal{B}$ that plays the Anon-SS-ID-CPA game with adversary $\mathcal{A}$ and utilizes the capability of $\mathcal{A}$ to break the $(\ell+1)$-MDDH assumption in $\mathbf{MG}_{\ell+1}(1^k)$. The constructed algorithm $\mathcal{B}$ in the Anon-SS-ID-CPA game is as follows.

\begin{itemize}
\item \textbf{Setup Phase}: Algorithm $\mathcal{B}$ gets as input an $(\ell+1)$-group system $\mathbf{MPG}_{\ell+1}$ and group elements $g,g^{x_1},\cdots,g^{x_{\ell+2}} \in \mathbb{G}_1 $ and $S \in \mathbb{G}_{\ell+1}$, where either $S= \hat{e}(g^{x_1},\cdots,g^{x_{\ell+1}})^{x_{\ell+2}}$ (i.e., $S$ is real) or $S\in\mathbb{G}_{\ell+1}$ uniformly (i.e., $S$ is random). $\mathcal{B}$ generates a $(q_p,2)$-MPHF $\mathbf{H}$ into $\mathbb{G}_\ell$, sets up the master public key as $\mathbf{PK}= (\mathbf{MPG}_{\ell+1},hk,\mathbf{H},h,h^\prime,\mathcal{ID},\mathcal{K},\mathcal{C})$ for $(h,h^\prime)= (g,g^{x_{\ell+1}})$ and $(hk,td)\leftarrow\mathbf{TGen}(1^k,g^{x_1},\cdots,g^{x_\ell},g)$, finally sends $\mathbf{PK}_{\mbox{\tiny IBKEM}}$ to adversary $\mathcal{A}$. Here, we use the \textbf{TGen} and \textbf{TEval} algorithms of  the $(q_p,2)$-MPHF property of $\mathsf{H}$.
 
\item \textbf{Query Phase 1}: Adversary $\mathcal{A}$ adaptively issues the following query multiple times.
\begin{itemize}
\item Decryption Key Query $\mathcal{Q}_{DK}^{IBKEM}(ID)$: Taking as input an identity $ID\in\mathcal{ID}_{\mbox{\tiny IBKEM}}$, algorithm $\mathcal{B}$ does the following steps:
	\begin{enumerate}
	\item Compute $\mathbf{TEval}(td,ID)=(a_{ID},B_{ID})$;
	\item If $a_{ID}=0$, return $\hat{S}_{ID}=\hat{e}(B_{ID},h^\prime)$; otherwise, abort and output $\bot$;
	\end{enumerate}
	Note that we have $\hat{S}_{ID}=\hat{e}(B_{ID},h^\prime)=\hat{e}(B_{ID},h)^{x_{\ell+1}}=\mathbf{H}_{hk}(ID)^{x_{\ell+1}}$. So $\mathcal{B}$ can answer a $\mathcal{Q}_{DK}^{IBKEM}(ID)$ query of $\mathcal{A}$ for identity $ID$ precisely when $a_{ID}=0$. 
 \end{itemize}

\item \textbf{Challenge Phase}: Adversary $\mathcal{A}$ sends two challenge identities $ID^*_0\in\mathcal{ID}_{\mbox{\tiny IBKEM}}$ and $ID_1^*\in\mathcal{ID}_{\mbox{\tiny IBKEM}}$ to algorihm $\mathcal{B}$;  $\mathcal{B}$ picks $\hat{d}\overset{\$}\leftarrow\{0,1\}$, and does the following steps:
\begin{enumerate}
\item Compute $\mathbf{TEval}(td,ID^*_0)=(a_{ID_0^*},B_{ID_0^*})$ and $\mathbf{TEval}(td,ID_1^*)=(a_{ID_1^*},B_{ID_1^*})$;
\item If $a_{ID_0^*}=0\bigvee a_{ID_1^*}=0$, then abort and output $\bot$;
\item Send the challenge key-and-encapsulation pair $(\hat{K}_{\hat{d}}^*=S^{a_{ID_{\hat{d}}^*}}\cdot \hat{e}(B_{ID_{\hat{d}}^*},g^{x_{\ell+1}},g^{x_{\ell+2}}),\hat{C}_0^*=g^{x_{\ell+2}})$ to adversary $\mathcal{A}$; 
\end{enumerate}

Suppose algorithm $\mathcal{B}$ does not abort (i.e., both $a_{ID_0^*}\neq 0$ and $a_{ID_1^*}\neq 0$ hold), we have $\mathbf{H}_{hk}(ID_0^*)=\hat{e}(g^{x_1},\cdots,g^{x_\ell})^{a_{ID_0^*}}\cdot \hat{e}(B_{ID_0^*},h)$ and $\mathbf{H}_{hk}(ID_1^*)=\hat{e}(g^{x_1},\cdots,g^{x_\ell})^{a_{ID_1^*}}\cdot \hat{e}(B_{ID_1^*},h)$.  Furthermore, if $S=\hat{e}(g^{x_1},\cdots,g^{x_\ell+1})^{x_{\ell+2}}$, we have
$\hat{K}_{\hat{d}}^*=S^{a_{ID_{\hat{d}}^*}}\cdot \hat{e}(B_{ID_{\hat{d}}^*},g^{x_{\ell+1}},g^{x_{\ell+2}})=\hat{e}(\mathbf{H}_{hk}(ID_{\hat{d}}^*),g^{x_{\ell+1}})^{x_{\ell+2}}$. This implies that the challenge key-and-encapsulation pair $(\hat{K}_{\hat{d}}^*,\hat{C}_0^*)$ is a valid one in this case. Otherwise, $\hat{K}_{\hat{d}}^*$ contains no information about $\hat{d}$.

\item \textbf{Query Phase 2}: This phase is the same as \textbf{Query Phase 2}. Note that in \textbf{Query Phase 1} and \textbf{Query Phase 2}, adversary $\mathcal{A}$ cannot query the decryption key corresponding to the challenge identity $ID_0^*$ or $ID_1^*$.

\item \textbf{Guess Phase}: Adversary $\mathcal{A}$ sends a guess $\hat{d}^\prime$ to algorithm $\mathcal{B}$. Let $\overline{Abort^\prime}$ denote the event that $\mathcal{B}$ does not abort in the previous phases. Let $\mathcal{I}=\{ID_1,\cdots,ID_{q_p},ID_0^*,ID_1^*\}$ be the set of the queried IDs by $\mathcal{A}$ and the challenge identities $ID_0^*$ and $ID_1^*$. Let $P_\mathcal{I}=Pr[\overline{Abort^\prime}|\mathcal{I}]$, which will be decided later. As in \cite{FHPS13,W05}, $\mathcal{B}$ ``artificially'' aborts with probability $1-1/(P_\mathcal{I}\cdot p(k))$ for the polynomial $p(k)$ from Definition \ref{D.MPHF}  and outputs $\bot$. If it does not abort, $\mathcal{B}$ uses the guess of $\mathcal{A}$. 
This means that if $\hat{d}=\hat{d}^\prime$, $\mathcal{B}$ outputs 1, otherwise it outputs 0.
\end{itemize}

In \textbf{Guess Phase}, $\mathcal{B}$ did not directly use the guess of $\mathcal{A}$, since event $\overline{Abort^\prime}$ might not be independent of the identities in $\mathcal{I}$. So $\mathcal{B}$ ``artificially'' aborts to achieve the independence. Let $\overline{Abort}$ be the event that $\mathcal{B}$ does not abort in the above game. We have that $Pr[\overline{Abort}]=1-Pr[Abort^\prime|\mathcal{I}]-Pr[\overline{Abort^\prime}|\mathcal{I}]\cdot (1-1/(P_\mathcal{I}\cdot p(k)))=1/p(k)$. Hence, we have $Pr[\mathcal{B}=1|S\ \text{is real}]=Pr[\overline{Abort}]\cdot(\frac{1}{2}+Adv^{\text{Anon-SS-ID-CPA}}_{IBKEM,\mathcal{A}})$ and $Pr[\mathcal{B}=1|S\ \text{is random}]=Pr[\overline{Abort}]\cdot\frac{1}{2}$, where $\frac{1}{2}+Adv^{\text{Anon-SS-ID-CPA}}_{IBKEM,\mathcal{A}}$ is the probability that $\mathcal{A}$ succeeds in the Anon-SS-ID-CPA game of IBKEM. Further, we have $$Pr[\mathcal{B}=1|S\ \text{is real}]-Pr[\mathcal{B}=1|S\ \text{is random}]=\frac{1}{p(k)}\cdot Adv^{\text{Anon-SS-ID-CPA}}_{IBKEM,\mathcal{A}}.$$ Hence, $\mathcal{B}$ breaks the $(\ell+1)$-MDDH assumption if and only if $\mathcal{A}$ breaks the Anon-SS-ID-CPA security of the above IBKEM scheme.

Finally, to evaluate $P_\mathcal{I}$, we can only approximate it (up to an inversely polynomial error, by running \textbf{TEval} with freshly generated keys sufficiently often), which introduces an additional error term in the analysis. We refer to \cite{W05} for details on this evaluation.
\end{proof}
\thispagestyle{empty}
\end{onecolumn}

\end{document}